\documentclass[10pt,doublecolumn,journal,letter]{IEEEtran}
\usepackage
[a4paper,% other options: a3paper, a5paper, etc
left=0.5in,
right=0.5in,
top=0.6in,
bottom=0.7in,
]{geometry}%\setlength{\columnsep}{0.23in}

\IEEEoverridecommandlockouts
\usepackage{array}
\usepackage{booktabs}
\usepackage{mathrsfs}
\usepackage{amsfonts}
\usepackage{amssymb}
\usepackage{amsmath}
\usepackage{graphicx}
\usepackage{multirow}
\usepackage{amsthm}
\usepackage{color}
\usepackage[table]{xcolor}
\usepackage{arydshln}
\usepackage{epstopdf}
\usepackage{url}
\usepackage{bm}
\usepackage{indentfirst}
\usepackage{epsfig}
\usepackage{xcolor}
\usepackage{algorithm}
\usepackage{algpseudocode}
\algdef{SE}[DOWHILE]{Do}{doWhile}{\algorithmicdo}[1]{\algorithmicwhile\ #1}%
\usepackage{cite}
\usepackage{booktabs}

\newtheorem{theorem}{Theorem}
\newtheorem{lemma}{Lemma}

\newtheorem{definition}{Definition}
\newtheorem{proposition}{Proposition}

\newtheorem{assumption}{Assumption}

\newtheorem{example}{Example}
\ifodd 0
\newcommand{\rev}[1]{{\color{blue}#1}}
\newcommand{\com}[1]{\textbf{\color{red} (Comment: #1) }}
\newcommand{\comg}[1]{\textbf{\color{blue} (COMMENT: #1)}}
\newcommand{\response}[1]{\textbf{\color{blue} (RESPONSE: #1)}}
\else
\newcommand{\rev}[1]{#1}
\newcommand{\com}[1]{}
\newcommand{\comg}[1]{}
\newcommand{\response}[1]{}
\fi

\newenvironment{conditions}
{\par\vspace{\abovedisplayskip}\noindent\begin{tabular}{>{$}l<{$} @{${}{}$} l}}
	{\end{tabular}\par\vspace{\belowdisplayskip}}
\def\N{\mathcal{N}}                 % user set

\def\K{\mathcal{K}}                 % content set
\def\E{\mathcal{E}}                 % edge set
\def\S{\mathcal{S}}

\def\KSize{L}

\def\Mcache{M^{ca}}

\def\Q{\boldsymbol{Q}}
\def\Qdown{Q^{down}}
\def\Qcpu{Q^{cpu}}
\def\Qup{Q^{up}}
\def\Qstore{\boldsymbol{Q}^{ca}}
\def\qstore{Q^{ca}}
\def\Qd2d{Q^{d2d}}

\def\edown{c^{down}}
\def\ecpu{c^{cpu}}
\def\eup{c^{up}}
\def\ed2d{c^{d2d}}

\def\D{\boldsymbol{D}}

\def\device{u}

\def\Din{\boldsymbol{D}^{in}}
\def\din{{D}^{in}}
\def\Dcpu{D^{cpu}}
\def\Dout{\boldsymbol{D}^{out}}

\def\Dup{\boldsymbol{D}^{up}}
\def\dup{{D}^{up}}
\def\Dstore{\boldsymbol{D}^{ca}}
\def\dstore{{D}^{ca}}

\def\xin{x^{in}}
\def\xdown{x^{down}}
\def\xup{x^{up}}
\def\xcpu{x^{cpu}}
\def\zin{z^{in}}
\def\zup{z^{up}}
\def\zstore{z^{ca}}
\def\x{x}

\def\Edown{E^{down}}
\def\Ecpu{E^{cpu}}
\def\Eup{E^{up}}
\def\Ed2d{E^{d2d}}

\def\Tcpu{T^{cpu}}
\def\Tup{T^{up}}
\def\Td2d{T^{d2d}}
\def\Tdown{T^{down}}

\def\TUdown{{\tau}^{down}}
\def\TUcpu{\tau^{cpu}}
\def\TUup{\tau^{up}}
\def\TUd2d{\tau^{d2d}}

\begin{document}

	%\title{A General Framework for User-Assisted Fog Computing on Communication, Computation, and Caching Sharing }
\title{Enabling Edge Cooperation in Tactile Internet via 3C Resource Sharing\vspace{-1mm}}%Edge Cooperation in Tactile Internet: Resource Sharing of Communication, Computation, and Caching}

\author{Ming Tang, Lin Gao, and Jianwei Huang\vspace{-8mm}
		\IEEEcompsocitemizethanks{
			\IEEEcompsocthanksitem		
			M. Tang and J. Huang (co-corresponding author) are with the Chinese University of Hong Kong, Hong Kong,
			E-mail: \{tm014, jwhuang\}@ie.cuhk.edu.hk;
			L. Gao (co-corresponding author) is with Harbin Institute of Technology, Shenzhen, China,
			E-mail: gaol@hit.edu.cn. 
			
Part of this work has been published at IEEE GLOBECOM \cite{globecom}.
			This work is supported by the General Research Funds (Project Number CUHK 14206315 and CUHK 14219016) established under the University Grant Committee of the Hong Kong Special Administrative Region, China, and the National Natural Science Foundation of China (Grant No. 61771162).}}

\maketitle

\maketitle
\begin{abstract}
Tactile Internet often requires (i) the ultra-reliable and ultra-responsive network connection and (ii) the proactive and intelligent actuation at edge devices. A promising approach to address these two requirements is to enable mobile edge devices to share their {communication}, {computation}, and {caching} (3C) resources via device-to-device (D2D) connections. In this paper, we propose a general 3C resource sharing framework, which includes many existing 1C/2C sharing  models in the literature as special cases. Comparing with 1C/2C models, the proposed 3C framework can further improve the resource utilization efficiency by offering more flexibilities in the device cooperation and resource scheduling. As a typical example, we focus on the energy utilization under the proposed 3C framework. Specifically, we formulate an energy consumption minimization problem under the 3C framework, which is an integer non-convex optimization problem. To solve the problem, we first transform it into an equivalent integer linear programming problem that is much easier to solve. Then, we propose a heuristic algorithm based on linear programming, which can further reduce the computation time and produce a result that is empirically close to the optimal solution. Moreover, we evaluate the energy reduction due to the 3C   sharing both analytically and numerically. Numerical results show that, comparing with the existing 1C/2C approaches, the proposed 3C sharing framework can reduce the total energy consumption by 83.8\% when the D2D energy is negligible. The energy reduction is still 27.5\% when the D2D transmission energy per unit time is twice as large as the cellular transmission energy per unit time.

\end{abstract}

\IEEEpeerreviewmaketitle

%!TEX root = Main-DASH-AUCTION.tex
%SourceDoc Main-DASH-AUCTION.tex

\vspace{-3mm}
\section{Introduction}\label{sec:intro}

\vspace{-1mm}

\subsection{Background and Motivation}
With the fast development of mobile communication and information technologies, Tactile Internet \cite{tactile2} has been recently proposed to support humans to control edge devices \rev{(e.g., robots, smart-phones, and virtual/augmented reality devices)} remotely in real time. By enabling tactile sensations, Tactile Internet can enrich the human-to-machine interaction and revolutionize the interaction among edge devices. Hence, Tactile Internet has attracted various applications such as automation industrial, real-time gaming, and virtual/augmented reality.

Many applications of Tactile Internet have two main requirements \cite{tactile2}: %require ultra-reliable and ultra-responsive real-time response. Such a requirement mainly consists of two aspects \cite{tactile2}: 
(i) the ultra-reliable and ultra-responsive network connection, and (ii) the proactive and intelligent actuation at edges. %Examples of such actuation can be edge computation and intelligent resource sharing. 
To address these two requirements from the perspective  of  mobile edge devices, a promising approach  is to enable  device-to-device (D2D) based resource sharing among edge devices, %\footnote{\rev{Although the D2D based resource sharing is at the expense of the D2D transmission cost, existing works have constructed demonstration systems and showed that such a cost could be compensated by the benefit obtained through the resource sharing\cite{exp1,exp2}.}} 
where the resources include communication, computation, and caching (3C) resources. \rev{Such a resource sharing can be realized by D2D communication technologies\cite{D2D}, including the ad hoc mode of the IEEE 802.11 standards, WiFi direct, and Bluetooth. Some business instances also support such a sharing, such as Open Garden (\url{https://www.opengarden.com/}).}

%Examples of the D2D based resource sharing in Tactile Internet are as follows. 
For example, to improve the reliability of network connections, edge devices can share their communication resources, in order to better utilize the devices' heterogeneous and fluctuating network capacities to satisfy their quality-of-service (QoS) requirements. To promote the intelligent actuation, edge devices can also share their computation and caching resources, so as to efficiently utilize their resources to satisfy their intensive task requirements.

Many of the existing works focused on the sharing of one  resource \cite{upn1,upn2,cloudlet1,cloudlet2,cache1,cache2}. For example, the user-provided networking models in \cite{upn1,upn2} focus on the sharing of communication resource, the ad hoc computation offloading models in \cite{cloudlet1,cloudlet2} focus on the sharing of computation resource, and the ad hoc content sharing models in \cite{cache1,cache2} focus on the sharing of caching resource. We refer to these models as 1C sharing models, since each of them focuses on one type of the 3C resources. Some other recent works further considered the sharing of two types of the 3C resources, which we call the 2C sharing models. Typical examples of 2C sharing  include the distributed data analysis models in \cite{fog1,fog3}, which focus on the sharing of  computation and caching resources.

\begin{figure} 
	\centering
	\includegraphics[height=3.4cm]{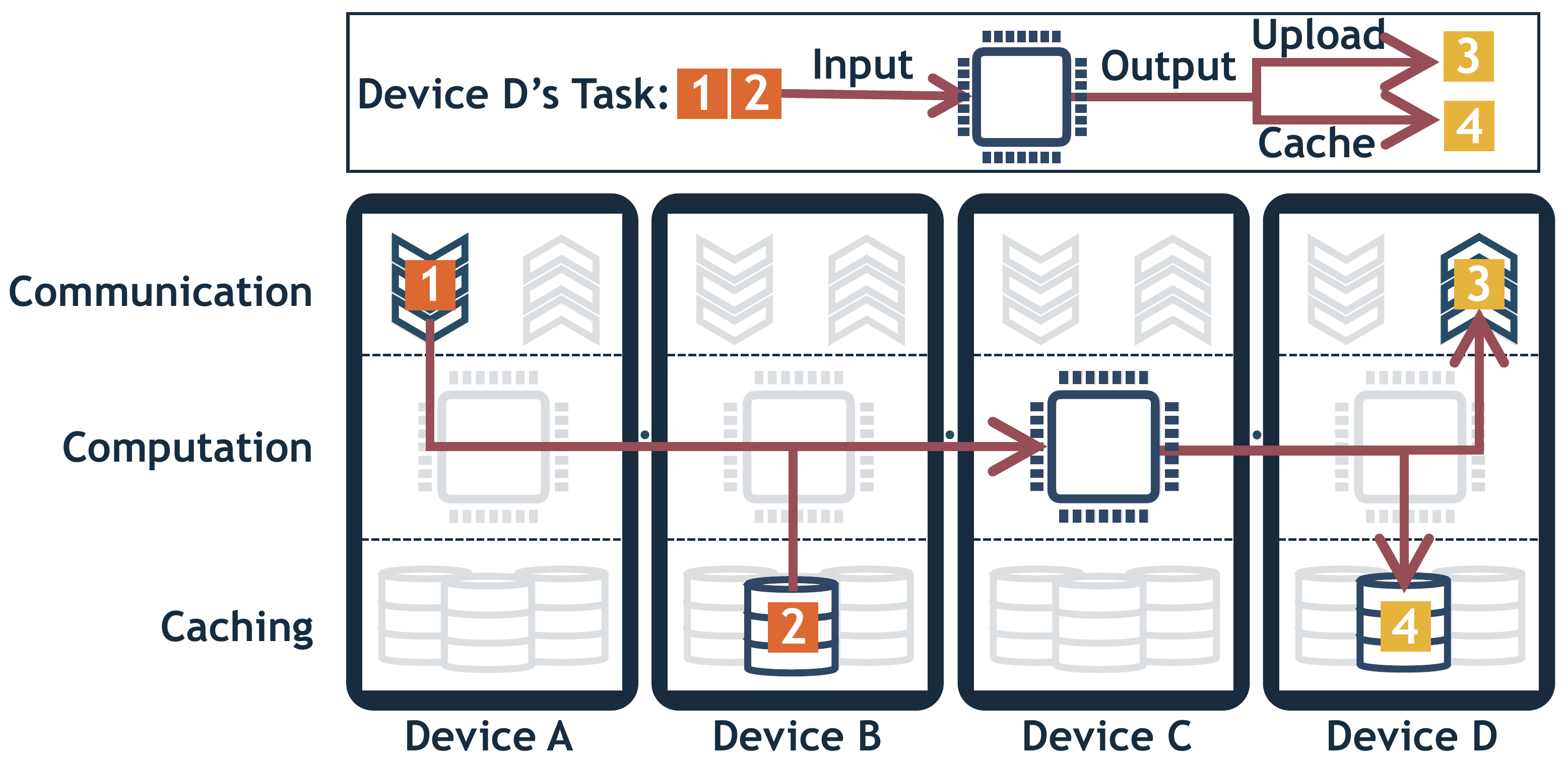}\vspace{-1mm}
	\caption{An example of the general 3C framework.}\label{fig:intro}
	\vspace{-4mm}
\end{figure}

Despite the success of the earlier 1C/2C resource sharing models, there are still significant potential benefits of exploiting the joint 3C resource sharing framework. Such a 3C sharing framework can further improve the resource utilization efficiency, by offering more  flexibilities in terms of device cooperation and resource scheduling.
Regarding the device cooperation,  a joint 3C framework can enable devices performing different tasks to cooperate with each other, which leads to an increased number of participating devices and hence more cooperation opportunities. Regarding the resource scheduling, the joint optimization of 3C resources can lead to a more efficient resource allocation.

\vspace{-2mm}
\subsection{Solution Approach and Contribution}

In this paper, we present the first study regarding the general 3C resource sharing framework%\footnote{\rev{Similar as existing 1C/2C models (e.g., \cite{upn1,upn2,cloudlet1,cloudlet2,cache1,cache2,fog1,fog3}), this framework is based on D2D communication technologies and infrastructure.}}
, which aims to  generalize existing edge device resource sharing (1C/2C) models and provide additional network design and optimization flexibilities. A key feature of this new 3C sharing framework is that it is centered around the characterization of the resource requirements of the tasks initialized by mobile devices, i.e., ``resource-centric", instead of emphasizing on the classification of these tasks, i.e., ``task-centric". In other words, any of the tasks (e.g., content retrieving, data analysis, or uploading) is modeled by the resources that  it requests, so that various types of  tasks requesting any combination of the 3C resources can coexist in the same framework.

Figure \ref{fig:intro} illustrates a simple example of the proposed 3C framework, where four devices \{A, B, C, D\} connect with each other via D2D and share their 3C resources to complete tasks. 
%\reva{which consists of four devices connecting with each other with D2D connections}. 
In this example, device D initializes a task that involves the following procedures: (i) retrieving contents ``1" and ``2" (either downloading from the Internet or fetching from some devices' caches), (ii) performing computation, and (iii) outputting contents ``3" (to the Internet) and ``4" (to device D's local cache). With the 3C framework, devices can share  their communication (downlink and uplink), computation (CPU), and caching resources. In this example, devices A and B are responsible for obtaining  the inputting contents and delivering them to device C for computation, then device C performs the computation and sends the outputting contents to device D, and finally device D further uploads content ``3" and caches content ``4" in its local cache.

To show the benefits of the 3C framework concretely, we focus on the energy consumption of mobile edge devices, and solve an energy consumption minimization problem under the 3C framework. Note that the proposed methodology can also be applied to other system optimization objectives (e.g., delay minimization or QoS maximization). A common feature of these optimization problems under the 3C framework is the introduction of integer variables due to the proposed  content-based framework (i.e, specifying the requested contents). %the consideration of caching sharing (e.g., who caches which content). 
The existence of the integer variables introduces difficulties in analyzing and solving the proposed problem. Moreover, tasks are often correlated with each other (e.g., due to the delays generated by the resource sharing), which further complicates the problem solving. We solve the energy minimization problem systematically and discuss the energy reduction due to the 3C sharing both analytically and numerically. Our key contributions are summarized as follows: 
\begin{itemize}
	\item \emph{General 3C Resource Sharing Framework:} We propose a general 3C sharing framework and the corresponding  ``resource-centric" mathematical formulation. This framework	 generalizes many existing 1C/2C resource sharing models, and improves the  resource utilization efficiency by encouraging more devices participating and more flexible resource scheduling. 
	\item \emph{Energy Efficiency Optimization:} We focus on the energy consumption of mobile edge devices under the 3C framework, and formulate and solve an energy consumption minimization problem. The problem is difficult as it is an integer non-convex optimization problem. We first transform it to an integer linear programming (ILP) problem, and then proposed a linear programming (LP) heuristic algorithm, whose output is empirically  close to the optimal solution.
	\item \emph{Theoretical Performance Analysis:} We analyze the energy consumption reduction due to the 3C resource sharing analytically.	We show that if the 3C framework can double the number of cooperative devices (comparing with 1C models), it can reduce the energy by a maximum of about 20\% of the energy consumed in noncooperation case (where devices do not cooperate)
	\item \emph{Simulation and Performance Evaluation:}	Comparing with existing 1C/2C sharing approaches, 3C sharing reduces the total energy by 83.8\% when  the D2D energy  consumption is negligible, and the energy  reduction is still  27.5\% when the  D2D energy per unit time is twice as large as the cellular energy per unit time. As for the computational complexity, when the network size is moderate (e.g., 27 devices), the heuristic algorithm reduces the computation time by 78.6\% at the cost of an optimality gap of 11.2\%.
\end{itemize}

The rest of this paper is organized as follows. Section \ref{sec:liter} reviews the related work, and Section \ref{sec:model} presents the  3C framework. In Section \ref{sec:algorithm}, we present the energy minimization problem transformation and heuristic algorithm design. We analyze the energy reduction due to the 3C framework theoretically and numerically in Section \ref{sec:theory} and Section  \ref{sec:experiment}, respectively. %In Section \ref{sec:theory}, we analyze the energy reduction due to the 3C framework. We then perform simulations on optimal and heuristic solutions comparison as well as 1C/2C and 3C approaches comparison in Section  \ref{sec:experiment}, 
We conclude in Section \ref{sec:conclude}.

\section{Literature Review}\label{sec:liter}

There are extensive studies working on the 1C/2C sharing models. Due to the limited space, we only briefly discuss some representative works that are most closely related to this study.

Most of the existing works considered 1C models. For example, Iosifidis \emph{et al.} \cite{upn1} and Syrivelis \emph{et al.} \cite{upn2} proposed user-provided network models for communication resource sharing, where nearby devices share their Internet connectivity for cooperative downloading. Militano  \emph{et al.} \cite{upload1} proposed an uploading resource sharing model, where devices form effective coalitions to share their uploading resources. Chi \emph{et al.} \cite{cloudlet1} and Chen \emph{et al.} \cite{cloudlet2} proposed ad hoc computation offloading models for computation resource sharing, where nearby mobile devices share their computation resources for data processing. Jiang \emph{et al.} \cite{cache1} and  Chen \emph{et al.} \cite{cache2} proposed ad hoc content sharing models for cached content sharing, where devices share their cached contents through D2D connections. 

Some recent works further considered 2C models. For example,  Stojmenovic \emph{et al.} \cite{fog1} and Destounis, \emph{et al.} \cite{fog3} considered distributed data analysis models, where some mobile devices share their cached data and other devices share their computation resources to process the shared data.

There are two limitations of the existing  1C/2C models: i) due to commonly adopted ``task-centric" approach, devices with different types  of tasks cannot cooperate (e.g., devices in user-provided network cannot cooperate with devices in ad hoc computation offloading), which restricts the pool of cooperative devices; and (ii) some tasks may request all of the 3C resources, which cannot be handled by   these existing 1C/2C models. In comparison, our proposed 3C framework addresses the above two limitations and improve  resource utilization efficiency  by providing more device cooperation and resource scheduling flexibilities.

\section{A General 3C Sharing Framework}\label{sec:model}

\subsection{System Model}\label{subsec:model-model} 
We consider three key elements in the 3C framework:  devices, tasks, and contents. 
\begin{itemize}
	\item \textbf{Device set} $\N=\{1,2,...,{N}\}$: The  devices form a mesh network (through D2D connections) for  cooperative task execution. For any device $n\in\mathcal{N}$, let $\E(n)$ denote the set of devices connected with device $n$ through D2D connections. Note that  $n\in\E(n)$.
	\item \textbf{Task  set} $\S=\{1,2,...,{S}\}$: The devices initialize these tasks, where a device may initialize one or more tasks.%, comprising set $\mathcal{S}$. 
	\item \textbf{Content  set} $\K=\{1,2,...,{K}\}$: The  contents can be the inputting or outputting of the tasks. %They can be downloaded from the Internet,  cached by the devices,  or produced by  computation. 
	A content $k\in\mathcal{K}$  has a size of $L_k$ (in bits). 
\end{itemize}
Next we provide detailed explanations of  devices  and tasks.
\subsubsection{\textbf{Device Model}} A device is characterized as follows, \rev{where the corresponding parameters are constants in a particular resource scheduling operation.}%\footnote{\rev{If the actual time for task execution is long, such constant capacities can be regarded as the average capacities over a relative long time period; if the actual time for task execution is short, such constant capacities can be regarded as the real-time capacities,  which can be approximated as constants.}}}
\begin{definition}[Device Model] Each device $n\in\N$ corresponds to a collection of tasks and resources, denoted by
	\begin{equation}\Q_n = (\boldsymbol{s}_n,\Qdown_n,\Qcpu_n,\Qup_n,\Qstore_n),
	\end{equation}
	where each notation  refers to a feature  of the device:
	\begin{conditions}
		\boldsymbol{s}_n    &-  the vector of  its initializing tasks' indexes, where\\
		&~~the dimension is the number of tasks it initializes,\\
		%	&~~\revc{}\\
		\Qdown_n&- downloading capacity (in bits per second),\\
		\Qcpu_n &-   computation capacity (in CPU cycles per second),\\
		\Qup_n    &-   uploading capacity (in bits per second),\\  
		\Qstore_n &- the vector of cached contents of dimension $K$, \\
		&~~where for any content $k\in\mathcal{K}$, %$\Qstore_n  =\{\qstore_{nk},\forall k\in\mathcal{K}\} $, where 
		$\qstore_{nk}=1$ if  $n$ has  \\
		&~~cached content $k$, and $\qstore_{nk}=0$ otherwise.
	\end{conditions}
	\noindent Let $\edown_n$, $\ecpu_n$, and $\eup_n$ denote the energy consumption of the downloading, computation, and uploading operations per unit second, respectively. 
\end{definition}

Next we define the model of the D2D connections.
\begin{definition}[Device-to-Device Model]
	For any two different devices $n,m\in\mathcal{N}$, let $\Qd2d_{n\rightarrow m}$ denote the D2D transmission capacity (bits per second) from device $n$ to device $m$, and let  $\ed2d_{n\rightarrow m}$ denote the D2D transmission energy per second.
\end{definition}

\subsubsection{\textbf{Task Model}}  Each task $s\in\S$ is represented by the task model shown in Figure \ref{fig:taskmodel}. Specifically, each task has a \emph{computation} module (which can have a zero computation requirement if the task does not involve any computation). \rev{The computation module requests some inputting contents, which can be downloaded from the Internet or fetched from devices' caches. The computation module then produces  some outputting contents, which can be uploaded to the Internet or cached at the task owner's device.} 
\begin{definition}[Task Model]
	Each task $s\in\S$ is denoted by 
	\begin{equation}\D_s = (\device_s, \Din_{s}, \Dcpu_s,\Dup_s, \Dstore_s),\end{equation}
	where each notation refers to a feature of the task:
	\begin{conditions}
		\device_s &- task owner (i.e., the device initializes this task),\\
		\Din_{s}&- the vector of the inputting contents,\\
		\Dcpu_s&- computation requirement (in total CPU cycles),\\
		\Dup_s&- the vector of the uploading contents,\\
		\Dstore_s&- the vector of the caching contents.\\	
	\end{conditions}
	\noindent All three  $\Din_{s}$, $\Dup_s$, and $\Dstore_s$ have the same size of $K$. For any $k\in\mathcal{K}$, $D_{sk}^{X}=1$ ($X\in in,up,ca$)  if   content $k$ is requested by task $s$ for inputting, uploading, or caching, respectively, and $D_{sk}^{X}=0$ otherwise.
\end{definition}

Each task can be divided into several subtasks:
\begin{definition}[Subtask]
	Each task consists of three subtasks: inputting, computation, and uploading\footnote{The contents to be cached, i.e., $\Dstore$, are cached at the task owner, so we do not regard it as a separate subtask.} %It is possible to consider the extension where the outputting content can be cached at other devices. This will add additional complexity of the model without significantly changing the results. We will not consider this due to space limit.} 
	subtasks.
\end{definition}
\rev{Our proposed  task model is content-based (i.e., specifying the requested contents) instead of data-based (i.e., only specifying the amount data needed). The content-based formulation makes the framework more flexible, as it only does not limit where to obtain the requested contents.} \rev{As a result, our task model can be used for modeling various applications \rev{whose  tasks can be broken down to input contents, computation, and output contents (or a subset of these three types of subtasks)}, such as  computation task offloading, communication task offloading, and cached content sharing. This is achieved by properly specifying the parameters in the task model, where the details are in Section \ref{subsec:generalization}.} %In addition, We consider a general 3C framework, where the inputting, computation, and uploading subtasks (of a same task) can be  performed by different devices. Moreover,  multiple  contents requested by the same  task can be inputted from or uploaded by different devices, leading to the maximum  resource sharing  flexibility. 

\begin{figure}[t]
	\centering
	\includegraphics[height=1cm]{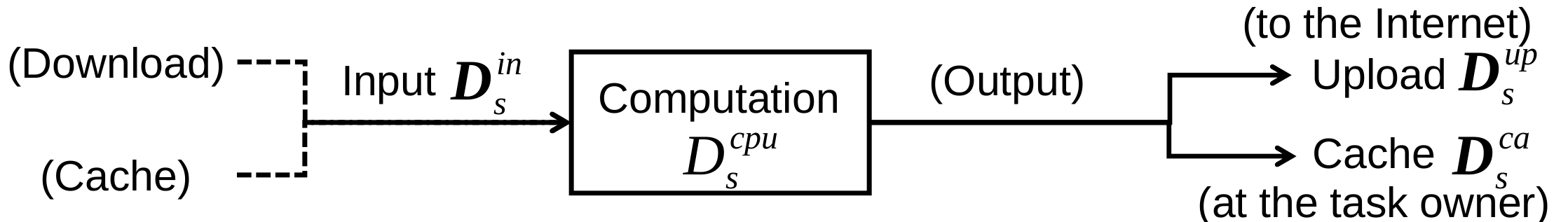}\vspace{-2mm}
	\caption{The task model.}\label{fig:taskmodel}\vspace{-2mm}
\end{figure}

\vspace{-2mm}
\subsection{Problem Statement}\label{sec:model-problem}
%{One key concern of mobile devices is their energy consumptions, hence as a concrete example, we focus on an  energy minimization under the 3C framework  in this paper. 
\rev{To demonstrate the proposed 3C framework, we focus on an energy minimization problem, as  mobile devices always concern about their energy consumption. Our  mathematical formulation is also  applicable to other optimization objectives (e.g., delay minimization or QoS maximization).}  

We first introduce decision variables, constraints, and  energy calculations. We then show the energy minimization problem.

\subsubsection{\textbf{Decision Variables}}
We consider a set of \emph{binary} decision variables with possible values from $\{0,1\}$ as follows. A variable equals  $1$ if its corresponding description is true, and equals $0$ otherwise.
\begin{conditions}
	\xin_{s, k\rightarrow n} &- device $n$ inputs task $s$' content $k$, \\
	\xdown_{s,k\rightarrow n} &- device $n$ downloads task $s$' content $k$, \\
	\xcpu_{s\rightarrow n} &- device $n$ performs task $s$' computation,\\
	\xup_{s, k\rightarrow n} &- device $n$ uploads task $s$' content $k$, \\
	\zin_{s, k\rightarrow i, j}&- task $s$' inputting content $k$ is delivered from $i$ to $j$,\\
	\zup_{s, k\rightarrow i, j}&- task $s$' uploading content $k$ is delivered from  $i$ to $j$,\\
	\zstore_{s, k\rightarrow i, j}&- task $s$' caching content $k$ is delivered from  $i$ to $j$.\\
\end{conditions}
\noindent Here $i$ and $j$ refer to  devices from set $\mathcal{N}$.

\subsubsection{\textbf{Constraints}} The system should satisfy the constraints as follows. %allocation constraints, capacity constraints, network flow balancing constraints, and worst case delay constraints as follows.

\textbf{\emph{Allocation constraints:}} Any task $s$' computation subtask should be allocated to only  one device:
\begin{equation}\label{constr:allo-cpu}
\textstyle{\sum_{n\in\mathcal{N}}\xcpu_{s\rightarrow n} = 1, ~\forall s\in\mathcal{S}.}\end{equation}
The inputting  of a requested content should be allocated to one device, i.e.,
\begin{equation}\label{constr:allo-in}
\textstyle{\sum_{n\in\mathcal{N}}\xin_{s,k\rightarrow n} =  \din_{sk},~\forall s\in\mathcal{S}, k\in\mathcal{K}.}\end{equation}
The  uploading of a requested content should be allocated to one device, i.e.,
\begin{equation}\label{constr:allo-up}
\textstyle{\sum_{n\in\mathcal{N}}\xup_{s,k\rightarrow n} =  \dup_{sk},~\forall s\in\mathcal{S}, k\in\mathcal{K}.}\end{equation}

\textbf{\emph{Capacity constraints:}} 
The device that is responsible for inputting a content should either has the content in its local cache or will download it from the Internet:
\begin{equation}\label{constr:cap-in}
\textstyle{\xin_{s,k\rightarrow n} \leq \qstore_{nk}+\sum_{s\in\mathcal{S}}\xdown_{s,k\rightarrow n},~\forall s\in\mathcal{S},n\in\mathcal{N}, k\in\mathcal{K}.}
\end{equation}
We consider the sum of the downloading subtask indicators, i.e., $\sum_{s\in\mathcal{S}}\xdown_{s,k\rightarrow n}$, as the device will own the content once it downloads content $k$  for any of the tasks $s\in\mathcal{S}$.

\textbf{\emph{Network flow balancing constraints:}} These constraints are related to the content delivery through multi-hop transmissions. \rev{For a content requested by a task, at any device, the incoming number of times that the device receives/generates the content should be the same as the number of times that the device transmits/consumes the content.} %For a particular content at any device, the incoming number of the copies of the content (i.e., the number of the copies of the content that it receives and generates) should be equal to the outgoing ones (i.e., the number of the copies of the content that it transmits and consumes). 

Taking the inputting content transmission variable $\zin_{s,k\rightarrow i, j}$ as an example. For any task $s\in\mathcal{S}$ and content $k\in\mathcal{K}$, the network flow balancing constraint at a device $i\in\mathcal{N}$ is 
\begin{equation}\label{constr:flow-in}
\sum_{j\in \E(i)} \zin_{s,k\rightarrow j,i} + \xin_{s,k\rightarrow i}\din_{sk} = \sum_{j\in \E(i)} \zin_{s,k\rightarrow i, j} + \xcpu_{s\rightarrow i}\din_{sk}.%,~\forall s\in\mathcal{S},k\in\mathcal{K},i\in\mathcal{N}.
\end{equation}
The left-hand side of \eqref{constr:flow-in} is the incoming number of times of task $s$' inputting content $k$, including the number of times  that device $i$ receives from its nearby devices and the number of times it generates for inputting (which equals one if $\xin_{s,k\rightarrow i}=1$ and $\din_{sk}=1$). The right-hand side of \eqref{constr:flow-in} is the outgoing number of times of task $s$' inputting  content $k$, including the number of times that device $i$ transmits to its nearby devices and the number of times it consumes to perform computation (which equals one if $\xcpu_{s\rightarrow i}=1$ and $\din_{sk}=1$).

Applying similar arguments to caching and uploading, we obtain the following constraints: 
\begin{equation}\label{constr:flow-cache}
\sum_{j\in \E(i)} \zstore_{s,k\rightarrow j,i} + \xcpu_{s\rightarrow i}\dstore_{sk}= 	\sum_{j\in \E(i)} \zstore_{s,k\rightarrow i,j} + \boldsymbol{1}_{i=u_s} \dstore_{sk},
\end{equation}
\begin{equation}\label{constr:flow-up}
\sum_{j\in \E(i)} \zup_{s,k\rightarrow j,i} + \xcpu_{s\rightarrow i}\dup_{sk} = 	\sum_{j\in \E(i)} \zup_{s,k\rightarrow i,j} + \xup_{s,k\rightarrow i}\dup_{sk}.
\end{equation}
Operator $\boldsymbol{1}_{i=u_s}=1$ if $i=u_s$, and $\boldsymbol{1}_{i=u_s}=0$ if $i\neq u_s$.

\textbf{\emph{Worst Case Delay Constraints:}} The worst case delay is the maximum delay that may happen due to the resource sharing. We first explain  the worst case delay constraints, and then compute the worst delays. 

First, to execute a task $s$, the \emph{worst} (maximum) delay of downloading\footnote{Inputting contents may come from downloading from the Internet or fetching from caches. For simplicity, we assume that fetching from caches is instantaneous, so we can focus on the delay caused by downloading.}, computation, and uploading subtasks, denoted by $T_s^{X}$ ($X\in\{down,cpu,up\}$), should be smaller than the corresponding delay bounds, respectively:
\begin{equation}\label{eq:delay}
\Tdown_s \leq \bar{T}^{down}_s, \Tcpu_s \leq \bar{T}^{cpu}_s, \Tup_s\leq \bar{T}^{up}_s,~\forall s\in\mathcal{S}.
\end{equation}
\rev{By using these separate delay constraints, we want to characterize the delay of each of the subtasks of a task.}  \rev{In addition, we ignore the D2D transmission delay  for simplification. In reality, such a delay is usually relatively small (e.g., Wi-Fi Direct has a transmission speed of up to 250Mbps\cite{WiFi-Direct}).} % For example, the requested contents have to be downloaded before certain deadline, beyond which the content may no longer be available at the corresponding server (e.g., a real-time video streaming application).}

Then, we show how  to calculate  the worst case delays ($\Tdown_s$, $\Tcpu_s$, and $\Tup_s$). The first step is to  calculate a device's time spending on completing all the subtasks allocated to it. Specifically, let $\TUdown_n$, $\TUcpu_n$, and $\TUup_n$ denote device $n$'s total time spending on completing all the downloading, computation, and uploading subtasks allocated to it, respectively:
\begin{equation}\label{eq:perform1}
\tau^{down}_n=\frac{\sum_{s=1}^{S}\sum_{k=1}^{K}x^{down}_{s,k\rightarrow n}\KSize_k}{Q^{down}_n},
\end{equation}
\begin{equation}\label{eq:perform2}
\tau^{up}_n=\frac{\sum_{s=1}^{S}\sum_{k=1}^{K}x^{up}_{s,k\rightarrow n}\KSize_k}{Q^{up}_n},
\TUcpu_n=\frac{ \sum_{s=1}^{S}\xcpu_{s\rightarrow n}\Dcpu_s}{\Qcpu_n}.
\end{equation}

{The second step is to calculate the worst case delays. Using \eqref{eq:perform1} as an example, we explain how to compute the  worst downloading delay $T_s^{down}$. We will first discuss the maximum downloading time that a device $n$ can impose on a task that it downloads for, and  then discuss how a task $s$ (requesting downloading from multiple devices) computes its worst downloading delay $T_s^{down}$. 
	
	For any device $n$, it may  download contents for multiple tasks, and the multiple tasks may be ready for downloading at different times.\footnote{The case of different ready times is more obvious for computation and uploading subtasks. Specifically, the computation of a task  is ready for execution only when the corresponding downloading has finished, and this time could be different for different tasks. Similar for uploading subtasks.}	 For simplicity, we assume that, if a device is downloading for multiple tasks at the same time, the device divides its {total downloading capacity} among the multiple tasks according to their total  downloading volumes.\footnote{\rev{We assume that when a subtask arrives at a device, it is executed without queuing, and it shares the capacity of the device with the other subtasks.}} For example, a device $n$ is downloading two contents (with sizes 1Mb and 2Mb, respectively) for task A, and one content (with a size 6Mb) for task B. Then, the device $n$ allocates $(1+2)/(1+2+6)=1/3$ and $6/(1+2+6)=2/3$ of its downloading capacity to task A and task B, respectively. Under this, the maximum downloading time that device $n$ can impose on a task is its total time spending on downloading, i.e., $\TUdown_n$. This happens when all the  tasks (allocated to device $n$) is ready for downloading at the same time, under which these tasks will share the device $n$'s capacity during the entire downloading process.

	For any task $s$, it  can obtain multiple contents from different devices. The worst downloading delay that task $s$ experiences  is the maximum downloading time $\TUdown_n$ among the device set $\{n|\sum_{k=1}^{K}\xin_{s,k\rightarrow n}(1-\qstore_{nk})> 0\}$. This set refers to the set of  devices that are responsible for task $s$' inputting but have not cached the contents yet (so they have to download the contents). Formally,
	\begin{equation}\textstyle{\label{eq:delay1}\Tdown_s = \max_{\{n|\sum_{k=1}^{K}\xin_{s,k\rightarrow n}(1-\qstore_{nk})> 0\}}\TUdown_n.} \end{equation}
	
	A similar idea applies to the calculation of the worst computation and uploading delays. Specifically, the worst computation delay  is the computation time of the device who performs task $s$' computation:
	\begin{equation}\label{eq:delay2}
	\textstyle{\Tcpu_s = \sum_{n=1}^{N}\xcpu_{s\rightarrow n} \TUcpu_n.}\end{equation}
	The worst uploading delay is the maximum uploading time  $\TUup_n$ among all the devices who perform task $s$' uploading: 
	\begin{equation}\label{eq:delay3}
	\textstyle{\Tup_s = \max_{\{n|\sum_{k=1}^{K}\xup_{s,k\rightarrow n}> 0\}}\TUup_n.} \end{equation}
	
	To clarify, these delay constraints are non-convex, since they contain quadratic forms that are not positive semidefinite, which makes it difficult to solve the  energy minimization problem (to be presented in Section \ref{subsubsec:problem}).

	\subsubsection{\textbf{Energy Calculations}}\label{subsubsec:energy} The energy for executing a task $s$ consists of the energy for downloading, computation, uploading, and D2D transmission. Formally,
	\begin{equation}E_s = \Edown_s  + \Ecpu_s +  \Eup_s+ \Ed2d_s.\end{equation}
	Each of these four terms is linear with the time that the devices spend on executing the corresponding operations:%\footnote{The linearity assumption has been adopted in the existing 1C/2C models considering energy consumption, e.g.,  \cite{upn1,upn2,cloudlet1,cloudlet2,upload1}.}
	\begin{equation*} \Edown_s = \sum_{n=1}^{N} \edown_n  \frac{\sum_{k=1}^{K}\xdown_{s, k\rightarrow n}\KSize_k}{\Qdown_n},\end{equation*}
	\begin{equation*}  \Ecpu_s = \sum_{n=1}^{N} \ecpu_n \frac{\xcpu_{s\rightarrow n}\Dcpu_s}{\Qcpu_n},\Eup_s = \sum_{n=1}^{N} \eup_n \frac{ \sum\limits_{k=1}^{K}\xup_{s,k\rightarrow n}\KSize_k}{\Qup_n},\end{equation*}
	\begin{equation*}\Ed2d_s = \sum_{i=1}^{N}\sum_{j=1}^{N}\ed2d_{i\rightarrow j} \frac{\sum\limits_{k=1}^{K}(\zin_{s,k\rightarrow i,j}+ \zup_{s,k\rightarrow i,j}+ \zstore_{s,k\rightarrow i, j})\KSize_k}{\Qd2d_{i\rightarrow j}}.
	\end{equation*}
	As an example, we explain task $s$' downloading energy  $\Edown_s$. It is the sum of the energy consumed by various devices for  downloading for task $s$, where each device's downloading energy  equals to the product of its energy coefficient and the downloading time.
	
	\subsubsection{\textbf{Problem Formulation}}\label{subsubsec:problem} We want to  minimize the energy consumption of the 3C framework under the proposed constraints. Formally, %we want to solve the following optimization problem: 
	\begin{equation*}
	\begin{aligned}
	& \underset{\boldsymbol{x},\boldsymbol{z}\in\{0,1\}}{\text{minimize}} & & \sum_{s\in\mathcal{S}} E_s\\
	& \text{subject to} & & \eqref{constr:allo-cpu}\sim\eqref{eq:delay}\\
	%& & &\boldsymbol{x},\boldsymbol{z}\in\{0,1\} \\
	\end{aligned}\eqno{\text{(OPT)}}\label{eq:problem1}
	\end{equation*}
	Problem $\text{(OPT)}$  is  non-convex due to the delay constraints. In Section \ref{sec:algorithm}, we transform it into an ILP problem, which can be solved by standard optimizers. We further propose a heuristic algorithm with a lower computational complexity.
	
	\begin{table}[t]
		\begin{center}
			\caption{Existing models that the 3C framework can generalize.}\label{table:model}\vspace{-2mm}
			\begin{tabular}{@{}ll@{}}
				\toprule
				Shared Resource & Examples and Task Models $\D_s$\\
				\hline
				\multirow{2}{*}{Downloading} & (a) User-provided networks (e.g., \cite{upn1,upn2})\\
				&~~~~$  (\device_s, \boldsymbol{D}^{data}_{s}, {0},\boldsymbol{0}, \boldsymbol{D}^{data}_{s})$ \\
				
				\multirow{2}{*}{Uploading}&(b) {Ad hoc content uploading} (e.g., \cite{upload1})\\
				&~~~~$ (\device_s,  \boldsymbol{D}^{data}_s, 0,\boldsymbol{D}^{data}_s,\boldsymbol{0})$\\
				
				\multirow{2}{*}{Content}& (c) {Ad hoc content sharing} (e.g., \cite{cache1,cache2})\\
				&~~~~$(\device_s, \boldsymbol{D}^{data}_s, 0,\boldsymbol{0}, \boldsymbol{D}^{data}_s)$\\
				
				\multirow{2}{*}{Computation}&(d) {Ad hoc computation offloading} (e.g.,  \cite{cloudlet1,cloudlet2})\\
				&~~~~$(\device_s,  \Din_{s}, \Dcpu_s, \boldsymbol{0},  \Dout_s)$\\
				
				\multirow{2}{*}{Hybrid}%&{Location-aware crowdsensing} & \cite{sense2,fog1}\\
				&(e) {Distributed data analysis}  (e.g., \cite{fog1,fog3})\\
				&~~~~$(\device_s, \Din_s, \Dcpu_s, \boldsymbol{0}, \Dout_s)$\\
				\bottomrule		
			\end{tabular}
		\end{center}
	\end{table}
	
	\subsection{Generalization of Existing Models in the Literature}\label{subsec:generalization}
	
	Through properly choosing  various parameters, the proposed 3C framework can generalize many of the existing 1C and 2C models. Examples are illustrated in Table \ref{table:model}. Among these models, the notation $\boldsymbol{D}^{data}_{s}$ (in (a), (b), and (c)) denotes the contents that are requested by the corresponding operations. 
	
	Figure \ref{fig:exampl0} illustrates the distributed data analysis model (e) as a special case  of our proposed 3C framework. Figure \ref{fig:exampl0}(a) corresponds to the model  in \cite{fog3}: two data source nodes S1 and S2 forward data to a computation node C for data analysis, then node C forwards the computation outputting to a destination D. Through specifying the task model as in Figure \ref{fig:exampl0} (b), our proposed model generalizes the model in \ref{fig:exampl0} (a), and can  achieve the optimal  resource allocation by solving  Problem $\text{(OPT)}$.

	\begin{figure}[t]
		\centering
		\centering
		\includegraphics[height=2.4cm]{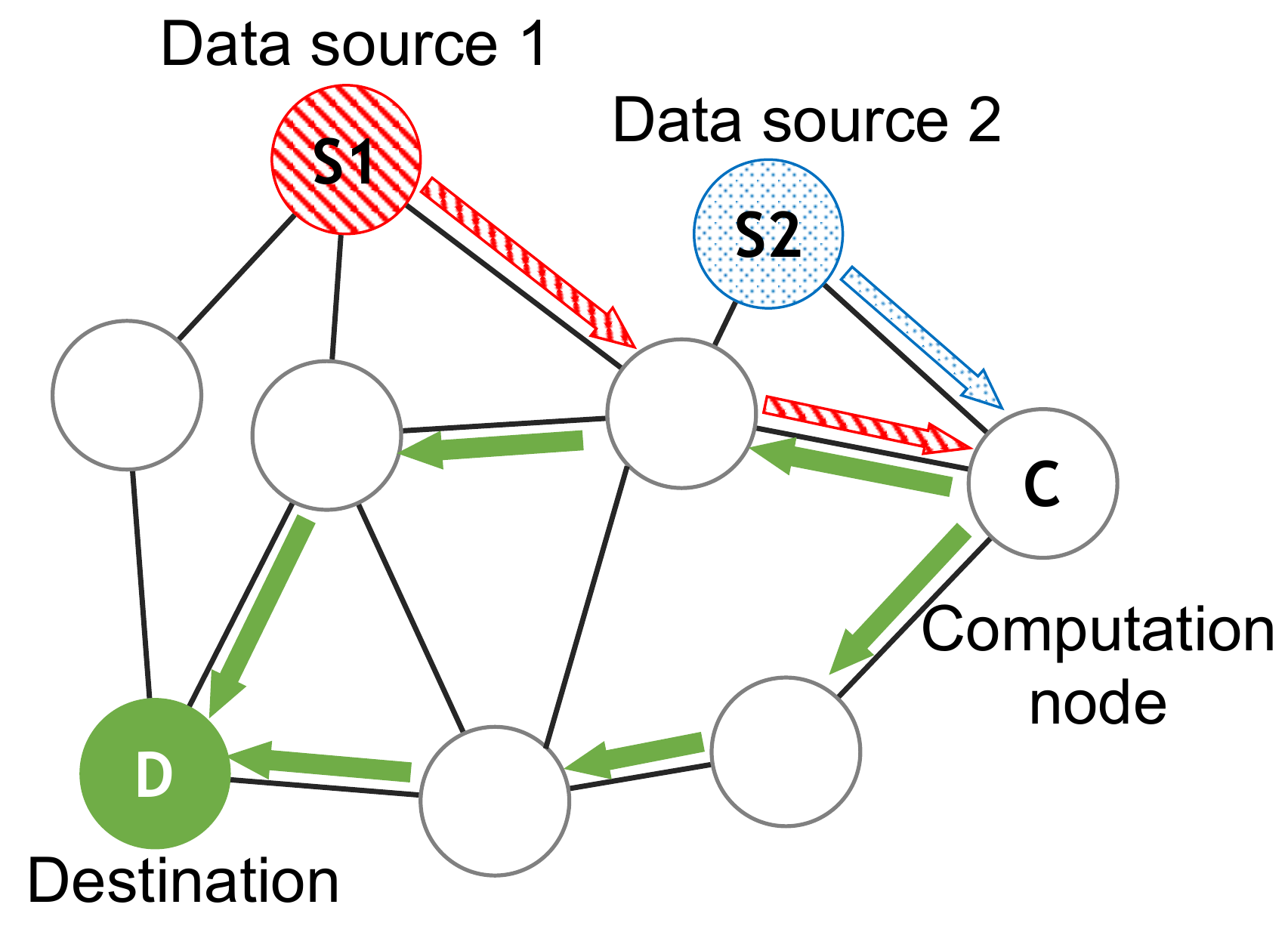}~~\includegraphics[height=2.4cm]{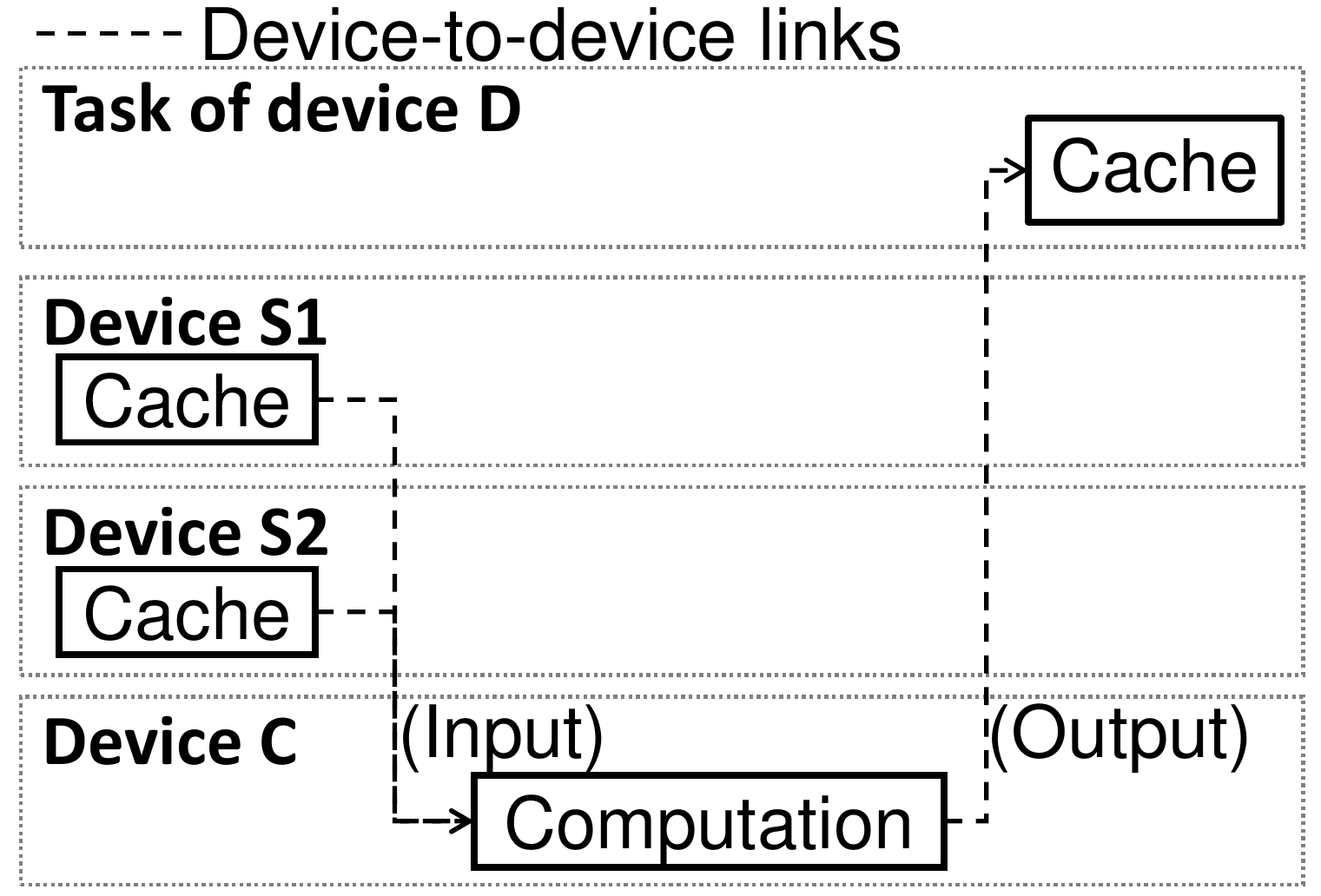}\\
		(a)~~~~~~~~~~~~~~~~~~~~~~~~~~~~~~~(b)\vspace{-2mm}\\
		\caption{An example of distributed data analysis: (a) existing model \cite{fog3}; (b)  the generalization by the 3C framework.}\label{fig:exampl0}\vspace{-2mm}
	\end{figure}
	
\section{Energy Minimization with  3C Sharing}\label{sec:algorithm}
In this section, we focus on the energy minimization problem $\text{(OPT)}$, which is an   integer non-convex optimization problem. To solve this problem, we first transform it into an ILP problem, which can be solved by standard optimizers. However, an ILP problem is an NP-complete problem (Theorem 18.1 in \cite{linearprogramming}), so  its computation time dramatically increases as the number of devices and tasks increases. Hence, we further propose a heuristic algorithm, which solves  a series of problems, each of which  is a LP relaxation (relaxing integer variables to continuous ones) of the original problem $\text{(OPT)}$. %This heuristic algorithm is guaranteed to produce an integer solution within the feasible region of the original problem $\text{(OPT)}$, and has a lower computational complexity. 
\rev{To clarify, the solutions in this section are based on a centralized scheduling, i.e., a centralized entity is needed to collect the information of devices and send control signals.\footnote{\rev{The centralized entity can be a  mobile device, an access point, or a base station. This entity can send the control messages to the rest of the network through both D2D connections and traditional uplink/downlink transmissions.}}} %\rev{Regarding the implementation, we focus on the online scheduling  scenario, where tasks arrive in real time. When Algorithm 1 converges in several milliseconds, it can be used to tackle a wide range of problems requiring real-time decisions. Even when Algorithm 1 cannot converge in several milliseconds, it may still be applied to problems whose resource allocation decision does not need to be made frequently.} % In future works, it is also interesting to focus on a distributed scheduling.} 

\subsection{Linear Transformation of Problem $\text{(OPT)}$}
We first  transform the integer non-convex  problem $\text{(OPT)}$ to an ILP problem. Since the non-convexity is mainly due to the delay constraints, the key focus will be how to transform the delay constraints to linear ones, after which Problem $\text{(OPT)}$  becomes an ILP.

The key transforming idea is as follows. Consider  a constraint in the form of $ \tau \times y\leq \bar{T}$, where the continuous variable $\tau \geq 0$, the discrete variable $y \in\{0,1\}$, and the fixed parameter $\bar{T} \geq 0$. We can transform the constraint into an equivalent linear form $ \tau - (1-y)\times M\leq \bar{T}$, where $M$ is {any  number that satisfies} $\tau-M\leq 0$. To see the equivalence of the two constraints, we consider two possible values of $y$. If $y=1$, both the original and the transformed constraints are $ \tau \leq \bar{T}$; if $y=0$, both constraints are true for any value of $\tau$. Hence, the two constraints are equivalent.

Next we use this key idea to explain the transformation of downloading delay constraints. The transformations of  computation and uploading delay constraints are similar. Finally, we present the transformed problem.

\subsubsection{\textbf{Transforming  downloading delay constraints}} We will transform the delay constraint ${T}^{down}_s\leq \bar{T}_s^{down}$ in \eqref{eq:delay} to a linear one, where ${T}^{down}_s$ is defined in $\eqref{eq:delay1}$. More specifically, this constraint can be written as 
\begin{equation}\label{eq:delay2-modi}
\TUdown_n  y_{s\rightarrow n}^{down}\leq \bar{T}_s^{down}, \forall s\in\mathcal{S}, n\in\mathcal{N},
\end{equation}
where $\tau^{down}_n$ is a linear function  of variables ${x}^{down}_{s,k\rightarrow n}$ as in \eqref{eq:perform1}. The variable  $y_{s\rightarrow n}^{down}\in\{0,1\}$ indicates  whether a task $s$' downloading is being allocated to device $n$,\footnote{If $y^{down}_{s\rightarrow n}=1$ (i.e., device $n$ downloads for task $s$), device $n$'s downloading time $\TUdown_n$ should satisfy the delay constraint, i.e., $\TUdown_n \leq \bar{T}_s^{down}$; if $y^{down}_{s\rightarrow n}=0$, device $n$'s downloading time  does not need to.} and satisfies the following conditions:  %for all $s\in\mathcal{S},n\in\mathcal{N},k\in\mathcal{K}$,
\begin{equation}\label{eq:y_down1}
\textstyle{y^{down}_{s\rightarrow n} \leq \min\{\sum_{k=1}^{K} \xin_{s,k\rightarrow n}(1-\qstore_{nk}), 1\},}
\end{equation}
\begin{equation}\label{eq:y_down2}
y^{down}_{s\rightarrow n} \geq  \xin_{s,k\rightarrow n}(1-\qstore_{nk}),\forall k\in\mathcal{K}.
\end{equation}
Specifically, when  $\xin_{s,k\rightarrow n}(1-\qstore_{nk})=0$ for all $k$ (i.e., device $n$ does not download any content for task $s$), we have $0\leq y^{down}_{s\rightarrow n}\leq0$, i.e., $y^{down}_{s\rightarrow n}=0$; otherwise, when there exists a  $k$ such that $\xin_{s,k\rightarrow n}(1-\qstore_{nk})=1$ (i.e., device $n$  downloads contents for task $s$), we have $0<\xin_{s,k\rightarrow n}(1-\qstore_{nk})\leq y^{down}_{s\rightarrow n}\leq 1$, i.e., $y^{down}_{s\rightarrow n}=1$.

Then, we transform constraints \eqref{eq:delay2-modi} based on the previously  mentioned transforming idea. We will choose a parameter $M_{n}^{down} $ that satisfies $\tau^{down}_n-M_{n}^{down} \leq 0$,\footnote{Note that $\tau^{down}_n$ is a linear function  of variables ${x}^{down}_{s,k\rightarrow n}$, which are unknown before solving the optimization problem. To ensure  the inequality $\tau^{down}_n-M_{n}^{down} \leq 0$ holds, we can set  $M_{n}^{down}=\max\{\bar{T}_s^{down},\forall s\},\forall n\in\mathcal{N}$. Similar ideas apply for the computation and uploading constraints.} and the constraint is given by
\begin{equation}\label{eq:M_down}
\tau^{down}_n- (1-y^{down}_{s\rightarrow n}) M_{n}^{down} \leq \bar{T}^{down}_s,\forall s\in\mathcal{S},n\in\mathcal{N},
\end{equation}
where $\tau^{down}_n\geq0$ is a linear function of ${x}^{down}_{s,k\rightarrow n}$ as in \eqref{eq:perform1}.%

\subsubsection{\textbf{Transforming computation and uploading delay constraints}} Based on similar ideas, for the computation delay constraint, we will choose a parameter $M_n^{cpu}$ that satisfies $\tau^{cpu}_n-M_n^{cpu} \leq 0$, and the equivalent constraint is given by
\begin{equation}\label{eq:M_cpu}
\tau^{cpu}_n - (1-\xcpu_{s\rightarrow n}) M_n^{cpu} \leq \bar{T}^{cpu}_{s}, \forall s\in\mathcal{S},n\in\mathcal{N},
\end{equation}
where $\xcpu_{s\rightarrow n}$ denotes whether device $n$ computes for task $s$.

For the uploading delay constraint, we will choose a parameter $M_n^{up}$ that satisfies $\tau^{up}_n-M_n^{up} \leq 0$, and the corresponding equivalent constraint is given by
\begin{equation}\label{eq:M_up}
\tau^{up}_n- (1-y^{up}_{s\rightarrow n})  M_n^{up} \leq \bar{T}^{up}_{s}, \forall s\in\mathcal{S},n\in\mathcal{N},
\end{equation}
where $	y^{cpu}_{s\rightarrow n} $ denotes whether device $n$ uploads for task $s$:
\begin{equation}\label{eq:y_up1}
\textstyle{y^{up}_{s\rightarrow n} \leq \min\left\{\sum_{k=1}^{K}\xup_{s,k\rightarrow n}, 1\right\}, y^{up}_{s\rightarrow n} \geq \xup_{s,k\rightarrow n},\forall k\in\mathcal{K}.}
\end{equation}

\subsubsection{\textbf{The Linear Transformation of Problem $\text{(OPT)}$}} Once replacing the delay constraint \eqref{eq:delay} with  \eqref{eq:y_down1}$\sim$\eqref{eq:y_up1}, we transform Problem $\text{(OPT)}$ into the following equivalent problem:
\begin{equation*}\label{eq:problem3}
\begin{aligned}
& \underset{\boldsymbol{x},\boldsymbol{z},\boldsymbol{y}\in\{0,1\}}{\text{minimize}} & & \sum_{s\in\mathcal{S}} E_s\\
& \text{subject to} & & \eqref{constr:allo-cpu}\sim\eqref{constr:flow-up}, \eqref{eq:y_down1}\sim\eqref{eq:y_up1}\\
%&&&T_s \leq \bar{T}_s, \forall s\in\S\\
%& \text{variables}& & \boldsymbol{x},\boldsymbol{z},\boldsymbol{y}\in \{0,1\} \\
\end{aligned}\eqno{(\text{OPT-LINEAR})}
\end{equation*}
Problem $\text{(OPT-LINEAR)}$  is an ILP, which can be solved by standard optimizers, e.g., Gurobi (\url{http://www.gurobi.com})}. 

%We want to emphasize again that the transformed delay constraints \eqref{eq:y_down1}$\sim$\eqref{eq:y_up2} are equivalent to the original constraints \eqref{eq:delay} if the newly introduced parameters are large enough. This is captured by the following Lemma \ref{lem:eq}. %${\tau}_{n}^{down} - M_{n}^{down}\leq 0$, ${\tau}_{n}^{cpu}-M_{n}^{cpu}\leq 0$, and ${\tau}_{n}^{up}-M_{n}^{up}\leq 0$.
%\begin{lemma}\label{lem:eq}
%	Problem $\text{(OPT-LINEAR)}$ is equivalent to Problem $\text{(OPT)}$ if the following is true for all $ n%\in\mathcal{N}$:
%	\begin{equation}\label{eq:M_condition}
%	{M}_{n}^{down}\geq {\tau}_{n}^{down}, {M}_{n}^{cpu}\geq{\tau}_{n}^{cpu},{M}_{n}^{up}\geq{\tau}_{n}^{up}.\end{equation}
%\end{lemma}
%\rev{The key proof idea is to show that when $M\geq \tau$, the inequalities $\tau\times y\leq \bar{T}$ and $\tau - (1- y)\times M \leq \bar{T}$ are equivalent for any $y=\{0,1\}$. Here, we omit the super- and sub-scripts of the variables and parameters.} 

%Note that ${\tau}_{n}^{down}$, ${\tau}_{n}^{cpu}$, and ${\tau}_{n}^{up}$ are functions  of decision variables, which are unknown before solving the optimization problem. To ensure that \eqref{eq:M_condition} holds, we can set the parameters $M_{n}^{down}$, $M_{n}^{cpu}$, and $M_{n}^{up}$ to be the maximum delay bounds  (i.e.,, $\max\{\bar{T}_s^{down},\forall s\}$, $\max\{\bar{T}_s^{cpu},\forall s\}$, and $\max\{\bar{T}_s^{up},\forall s\}$, respectively, for any $n\in\mathcal{N}$).

However, directly solving Problem $\text{(OPT-LINEAR)}$ using standard optimizers works well when the network size (e.g., the number of devices and tasks) is reasonably small. %\footnote{In simulations  in Section \ref{sec:experiment}, directly solving Problem $\text{(OPT-LINEAR)}$ has a  low computation time when the device number is smaller than 20.} 
As the system size increases, the computation time dramatically increases, as Problem $\text{(OPT-LINEAR)}$ (an ILP)  is   NP-complete\cite{linearprogramming}. To address this complexity issue, we propose a  heuristic algorithm based on the original Problem $\text{(OPT)}$. 

\subsection{A Heuristic Algorithm of Solving Problem $\text{(OPT)}$}

The key idea is to iteratively solve a series of modified versions of Problem $\text{(OPT)}$, where we remove  the delay constraints and relax the integer variables to continuous ones (i.e., LP relaxation\cite{linearprogramming}, so that the modified problems are LP problems). At the end of each iteration, the algorithm will check whether the removed delay constraints  are satisfied. If not, the algorithm will prevent  some tasks from being allocated to certain devices (in order to address the violated delay constraints), and solve a new version of  modified problem.  The algorithm iterates  until all the delay constraints are satisfied. Note that we do not check whether the variables satisfy the integer constraints in the algorithm. Later we will show that, however, the algorithm is guaranteed to produce integer solutions that are feasible for Problem $\text{(OPT)}$. 

Next we first describe the modified  problem, then we propose the heuristic algorithm.

\subsubsection{\textbf{A Modified Problem of Problem $\text{(OPT)}$}} 
Comparing with the original Problem $\text{(OPT)}$, the modification involves removing delay constraints, relaxing integer variables, and adding control parameters that can prevent certain tasks from being allocated to certain devices. We first introduce the control parameters, then propose the modified problem.

In order to prevent particular subtasks from being allocated to certain devices, we introduce the following binary control parameters $\tilde{N}^{in}_{s,k\rightarrow n}$, $ \tilde{N}^{cpu}_{s\rightarrow n}$, and $ \tilde{N}^{up}_{s,k\rightarrow n}$:
\begin{equation}\label{constr:control-in}
x^{in}_{s,k\rightarrow n}\leq\tilde{N}^{in}_{s,k\rightarrow n},~
\x^{cpu}_{s\rightarrow n}\leq \tilde{N}^{cpu}_{s\rightarrow n}, ~x^{up}_{s,k\rightarrow n}\leq \tilde{N}^{up}_{s,k\rightarrow n}.
\end{equation}
Take $\tilde{N}^{in}_{s,k\rightarrow n}$ as an example: if it equals zero, then \eqref{constr:control-in} indicates that $ x^{in}_{s,k\rightarrow n}$ can only be zero, so the content $k$ of task $s$ cannot be allocated to device $n$; if it equals one, then $ x^{in}_{s,k\rightarrow n}$ can be either zero or one, so the allocation is not prevented. The same idea applies to $\tilde{N}^{cpu}_{s\rightarrow n}$ and $\tilde{N}^{up}_{s,k\rightarrow n}$.

Then, we introduce the modified problem of Problem $\text{(OPT)}$ by removing  delay constraints \eqref{eq:delay}, relaxing integer variables (i.e., relaxing $\boldsymbol{x},\boldsymbol{z}\in \{0,1\}$ to be $\boldsymbol{x},\boldsymbol{z}\in [0,1]$),  and adding control constraints $\eqref{constr:control-in}$:
\begin{equation*}\label{eq:problem4}
\begin{aligned}
& \underset{\boldsymbol{x},\boldsymbol{z}\in [0,1]}{\text{minimize}} & & \sum_{s\in\mathcal{S}} E_s\\
& \text{subject to} & & \eqref{constr:allo-cpu}\sim\eqref{constr:flow-up},\eqref{constr:control-in}%\sim\eqref{constr:control-up}
\\
%& \text{variables}& & \boldsymbol{x},\boldsymbol{z}\in [0,1] \\
\end{aligned}\eqno{\text{(OPT-RELAX)}}
\end{equation*}
To clarify, there are many versions of Problem $\text{(OPT-RELAX)}$, each of which corresponds to a set of parameter choices of $\tilde{N}^{in}_{s,k\rightarrow n}$, $ \tilde{N}^{cpu}_{s\rightarrow n}$, and $ \tilde{N}^{up}_{s,k\rightarrow n}$. Moreover, Problem $\text{(OPT-RELAX)}$ is an LP problem, which can be solved using various methods such as Simplex method\cite{simplex}.

\begin{algorithm}[t]
\caption{Heuristic Algorithm}
\label{alg:suboptimal}
\begin{algorithmic}[1]
	\State \textbf{Initialization:} $\tilde{\boldsymbol{N}}^{in},\tilde{\boldsymbol{N}}^{cpu},\tilde{\boldsymbol{N}}^{up}\leftarrow 1$
	\State $\boldsymbol{x},\boldsymbol{z}\leftarrow$Solve Problem $\text{(OPT-RELAX)}$ (e.g., using Simplex method \cite{simplex})
	\State Calculate $T_s^{down}, T_s^{cpu}, T_s^{up}$ using \eqref{eq:delay1}, \eqref{eq:delay2}, \eqref{eq:delay3}
	\While{delay constraints \eqref{eq:delay} are not fully  satisfied}
	\For{each task $s\in\S$}
	\If{$\Tdown_s > \bar{T}^{down}_s$}
	\State  $\hat{n}\leftarrow \max_n\{\tau^{down}_ny^{down}_{s\rightarrow n}\}$ \Comment{select a device} 
	\State $\boldsymbol{I}\leftarrow\{s|\!\sum  \limits_{\hat{s}\in\boldsymbol{s}_{\hat{n}}}\sum \limits_{k}\xin_{\hat{s},k\rightarrow \hat{n}}\xin_{s,k\rightarrow \hat{n}}\geq1\}$
	
	\Comment{find device $\hat{n}$'s \emph{non-preventable} task set}
	\State $\bar{s} \leftarrow \arg\min_{s}\{\bar{T}^{down}_s| y^{down}_{s\rightarrow \hat{n}}\neq 0 , {s}\notin\boldsymbol{I}\}$ 
	
	\Comment{select a task $\bar{s}$ that is preventable}
	\State $\bar{\mathcal{K}}\leftarrow \{k|\xin_{\bar{s},k\rightarrow \hat{n}}(1-\Qstore_{\hat{n}k})=1\}$
	
	\Comment{select contents}
	\State $\tilde{N}^{in}_{{s},k\rightarrow\hat{n}} \leftarrow 0, \forall\{s,k|\xin_{s,k\rightarrow\hat{n}}=1,k\in\bar{\mathcal{K}}\}$
	
	\Comment{prevent the allocations}
	\EndIf
	\If{$\Tcpu_s > \bar{T}^{cpu}_s$}
	\State  $\hat{n}\leftarrow \max_n\{\tau^{cpu}_ny^{cpu}_{s\rightarrow n}\}$ \Comment{select a device}
	\State $\bar{s} \leftarrow \arg\min_{s}\{\bar{T}^{cpu}_s| y^{cpu}_{s\rightarrow \hat{n}}\neq 0 , s\notin\boldsymbol{s}_{\hat{n}}\}$
	
	\Comment{select a  task}
	\State $\tilde{N}^{cpu}_{\bar{s}\rightarrow\hat{n}} \leftarrow 0$\Comment{prevent the allocation}
	\EndIf
	\If{$\Tup_s > \bar{T}^{up}_s$}
	\State  $\hat{n}\leftarrow \max_n\{\tau^{up}_ny^{up}_{s\rightarrow n}\}$\Comment{select a device}
	\State $\bar{s} \leftarrow \arg\min_{s}\{\bar{T}^{up}_s| y^{up}_{s\rightarrow \hat{n}}\neq 0 , s\notin\boldsymbol{s}_{\hat{n}}\}$
	
	\Comment{select a  task}
	\State $\tilde{N}^{up}_{\bar{s},k\rightarrow\hat{n}} \leftarrow 0$,$\forall k$\Comment{prevent the allocations}
	\EndIf
	\EndFor
	\State $\boldsymbol{x},\boldsymbol{z}\leftarrow$Solve Problem  $\text{(OPT-RELAX)}$ (e.g., using Simplex method \cite{simplex}) %using Simplex method\cite{simplex}
	\State Calculate $T_s^{down}, T_s^{cpu}, T_s^{up}$ using \eqref{eq:delay1}, \eqref{eq:delay2}, \eqref{eq:delay3}
	\EndWhile
	\State \Return $\boldsymbol{x},\boldsymbol{z}$
\end{algorithmic}
\end{algorithm}

\subsubsection{\textbf{A Heuristic Algorithm  to solve Problem $\text{(OPT)}$}}
The heuristic algorithm will iteratively solve multiple versions of Problem $\text{(OPT-RELAX)}$ as follows. At the beginning, no allocation is prevented, i.e., $\tilde{N}^{in}_{s,k\rightarrow n} = \tilde{N}^{cpu}_{s\rightarrow n} =  \tilde{N}^{up}_{s,k\rightarrow n} =1$ for all $s$, $k$, and $n$. We first optimize the corresponding Problem  $\text{(OPT-RELAX)}$ (e.g., using Simplex method \cite{simplex}), and check whether the optimal solution satisfies the delay constraints in \eqref{eq:delay}. If yes, then the obtained optimal solution of solving Problem  $\text{(OPT-RELAX)}$ is the optimal solution of Problem $\text{(OPT)}$; if not, we need to revise the  control parameters in Problem $\text{(OPT-RELAX)}$ (i.e., setting some $\tilde{N}^{in}_{s,k\rightarrow n}$, $\tilde{N}^{cpu}_{s\rightarrow n}$, or $\tilde{N}^{up}_{s,k\rightarrow n}$ to be zeros), with details discussed in the next paragraph. We optimize  Problem $\text{(OPT-RELAX)}$ iteratively until obtaining a solution that satisfies  the delay constraints in \eqref{eq:delay}. The  algorithm is given in Algorithm \ref{alg:suboptimal}. 

We now discuss how the algorithm chooses the proper version of Problem $\text{(OPT-RELAX)}$ to solve by setting $\tilde{N}^{in}_{s,k\rightarrow n}$, $\tilde{N}^{cpu}_{s\rightarrow n}$, or $\tilde{N}^{up}_{s,k\rightarrow n}$ for inputting (in lines 6-12  of Algorithm \ref{alg:suboptimal}),\footnote{The inputting subtask prevention corresponds to  downloading delays, because the downloading is the operation inducing inputting delays.} computation (in lines  13-17), and uploading (in lines  18-22) subtasks, respectively. We first introduce the general idea, then explain a special setting  of the inputting subtask. %Regarding each of these tasks, T

The general idea of preventing some allocations for the inputting, computation, and uploading subtasks is as follows. For a particular subtask of a task $s$, if its corresponding delay (after solving a version of Problem $\text{(OPT-RELAX)}$) is larger than the delay bound,  then  the algorithm will (i)  find the device $\hat{n}$ that induces the maximum delay, (ii) at device $\hat{n}$, find the task $\bar{s}$ with the tightest delay bound among all the tasks that are allocated to device $\hat{n}$ (excluding device $\hat{n}$'s tasks), %\footnote{This is used to ensure that  the heuristic algorithm always has an output, which is the noncooperation case (where each device performs its tasks on its own). This will be discussed in the next subsection.} 
(iii) prevent task $\bar{s}$ from being allocated to device $\hat{n}$ by setting the corresponding  $\tilde{N}^{in}_{s,k\rightarrow n}$, $\tilde{N}^{cpu}_{s\rightarrow n}$, or $\tilde{N}^{up}_{s,k\rightarrow n}$ to be zeros.

Next we discuss a special setting of the inputting (downloading) subtask. Specifically, device $\hat{n}$ may download the same content for both itself and other devices, but it  should not prevent  the content downloading of its own tasks. So we define a \emph{non-preventable} set $\boldsymbol{I}$ (in line  8), which contains the tasks that request a same content as device $\hat{n}$ does. Only the tasks outside the non-preventable set can be prevented (in line 9). In addition, only the contents that have not be cached (have to be downloaded) are prevented (in lines 10 and 11), because cached contents do not induce delays.

\subsubsection{\textbf{Properties of Algorithm \ref{alg:suboptimal}}}
We first make  an assumption  that is often satisfied in practice, and then characterize several properties of the proposed heuristic algorithm, including its feasible solution output guarantee, its performance guarantee, and its computation complexity.

\begin{assumption}[Feasible Noncooperation Case]\label{ass:noncooperation}
Noncooperation (i.e., each of the devices performs its tasks on its own) is within the feasible region of Problem $\text{(OPT)}$.
\end{assumption}
This assumption implies that each device is capable of executing its tasks on its own. If Assumption \ref{alg:suboptimal} is violated, some tasks may become infeasible to complete, as cooperation is not always guaranteed in practice. 

Under Assumption \ref{ass:noncooperation}, \rev{Algorithm \ref{alg:suboptimal} is guaranteed  to converge and output an integer feasible solution of Problem $\text{(OPT)}$.}
\begin{proposition}[Guarantee of Feasible Output]\label{prp:1}  Algorithm \ref{alg:suboptimal} is guaranteed to \rev{converge and} produce an integer solution that is within the feasible region of Problem $\text{(OPT)}$.
\end{proposition}
The proof is given  in Appendix \ref{app:prp-1}.  \rev{Specifically, to prove the convergence, we have to show that the noncooperation solution will never be excluded from the feasible region of Problem $\text{(OPT-RELAX)}$, so the algorithm will definitely converge when reaching the noncooperation solution (if it has not converged earlier).  To prove the integer feasible solution, we have to prove that the optimal solution of any version of Problem $\text{(OPT-RELAX)}$ is always integer (as its matrix of constraint coefficients is totally unimodular\cite{unimodular}) and within the feasible region of Problem $\text{(OPT)}$.} %, even though there is no integer constraint in that problem.
%The integer is proved by showing that  Problem $\text{(OPT-RELAX)}$ is a linear programming whose matrix of constraint coefficients is totally unimodular\cite{unimodular}.} 
%\rev{To clarify, Proposition \ref{prp:1} implies that although the key part that we relax (in the heuristic algorithm) is the integer variables, %, where the relaxation is used to reduce the computation time (transforming an ILP problem $\text{(OPT)}$ to an LP problem $\text{(OPT-RELAX)}$). 
% such an LP relaxation does not affect the integer output of the heuristic algorithm.} %, as we can prove that solving $\text{(OPT-RELAX)}$ with $\boldsymbol{x},\boldsymbol{z}\in[0,1]$ (e.g., using Simplex method) is guarantee to output an integer solution, i.e., $\boldsymbol{x},\boldsymbol{z}\in\{0,1\}$. 

We now show the performance guarantee of Algorithm \ref{alg:suboptimal}.
\begin{proposition}[Performance Guarantee]\label{prp:2}  \rev{(i)} The energy consumption of the heuristic algorithm output is no larger than that of the noncooperation case \rev{(i.e., each of the devices performs its tasks on its own)}. \rev{(ii)} When there is no delay constraint,  the heuristic algorithm output is an optimal solution of the original problem $\text{(OPT)}$.
\end{proposition}
The proof is given  in Appendix \ref{app:prp-2}. Specifically, \rev{statement (i) is proved by showing that the noncooperation case is always within the feasible region of Problem (OPT-RELAX), so the optimal solution of  Problem (OPT-RELAX) is always no worse than that is achieved under the noncooperation case.} \rev{Statement (ii) is proved by showing that when there is no delay constraint,  Problem (OPT) is equivalent to the initial version of Problem (OPT-RELAX) without eliminating any feasible solution. In this case, Algorithm \ref{alg:suboptimal} will terminate at the first iteration and output the optimal solution to Problem (OPT).}  We will further evaluate the performance of this heuristic algorithm under the settings with delay constraints in Section \ref{subsec:compare}.

Regarding the  complexity of Algorithm \ref{alg:suboptimal}, its maximum iteration time  is as follows:
\begin{proposition}[Maximum Iteration Time]\label{prp:3} The maximum iteration time of this heuristic algorithm is $S\times (N-1)$, where $S$ is the task number and $N$ is device number. 
\end{proposition}
The proof is given in Appendix \ref{app:prp-3}. \rev{The key idea is as follows. There are $S\times (N-1)$ possible allocations from a task to a device other than the task owner. The algorithm terminates no later than the iteration where all possible allocations are removed from the feasible region of Problem (OPT-RELAX). This is because  we have shown that the algorithm would definitely converge if reaching the noncooperation solution. Since each iteration removes at least one of such allocations from the feasible region, the algorithm terminates with a maximum $S\times (N-1)$ iterations.} 

Recall that  Problem $\text{(OPT)}$ is NP-complete, which cannot be solved in polynomial time in general. In comparison, \rev{Proposition \ref{prp:3} shows that the heuristic algorithm has a maximum of $S\times (N-1)$ iterations, each of which solves an LP problem  $\text{(OPT-RELAX)}$ that is a P-complete problem. This implies that the heuristic algorithm can terminate in polynomial time.}

\section{Energy Reduction Due to 3C Sharing}\label{sec:theory}

The proposed 3C framework is ``resource-centric" instead of ``task-centric", so that it provides additional flexibilities in terms of device cooperation. More specifically, it promotes cooperation opportunities through enabling devices performing different tasks to cooperate. In this section, we study how much a 3C framework can reduce the energy consumption  through a specific problem setting, comparing with 1C models. We first introduce system settings, then discuss the energy reduction due to  the 3C framework.

\subsection{System Settings}\label{subsec:analysis-setting}
In order to  derive the closed-form solutions of the energy reduction, we consider specific device and task models as follows. We consider a random graph model $G(N, p)$\cite{networkscience}, where there are $N$ devices in the graph and every two devices  are connected randomly and independently  with a probability $p$. Suppose that the network is large and sparse, so that $N$ approaches infinite with $Np$ being a constant. %\footnote{This assumption is consistent with  the phenomenon that most real networks are sparse, i.e., the number of devices that a device may connect with is significantly smaller than the total number of devices\cite{networkscience}.} 
These devices initialize a set of tasks. Since we focus on the comparison between 1C models and 3C framework, we assume that each task only needs one of the 3C resources.

The devices are heterogeneous in terms of their owned resources. Specifically, each device $n$ owns some resources $\Qdown_n$, $\Qcpu_n$, and $\Qup_n$. The capacities $Q^{X}_n$ ($X\in\{down,cpu,up\}$) is independent and identically distributed (i.i.d.) with the cumulative distribution function (cdf) $F_Q^{X}(x)$ and the probability density function (pdf) $f_Q^{X}(x)$. The support of capacity $Q_n^X$ is  $(\underline{Q}^{X},  \overline{Q}^{X})$, hence  $F_Q^{X}(\underline{Q}^{X}) =0$  and $F_Q^{X}(\overline{Q}^{X}) =1$. For the convenience of analysis, we assume that the energy coefficients of the devices are homogeneous, i.e., $c^X_n = c^X,X \in \{down, cpu, up\}$, $\forall n$. In addition, each device $n$ uniformly and randomly caches $\Mcache$ contents in its cache, %\footnote{The uniformly and randomly caching is a widely used benchmark in proactive caching\cite{proactivecahce}, leading to a performance that is no better than using optimal caching strategies.}
i.e., $\sum_{k=1}^{K}\qstore_{nk}=\Mcache$ for all $n$. For  simplification, we assume that all the contents have the same size that is  normalized to one, i.e., $\KSize_k=1$, for each $k$, \rev{similar as in existing caching studies on performance analysis (e.g., \cite{caching-size}).}

We aim to study the system under the general distribution function, which is quite challenging to do. Hence, we further make the following simplifying assumption for the rest of Section \ref{sec:theory}. %These assumptions lead to a larger energy reduction that the framework can achieve (when considering single-hop cooperations), comparing with relaxing these assumptions. 
A more realistic case (with these assumptions relaxed) is evaluated empirically in Section \ref{sec:experiment}.
\begin{assumption}\label{ass:analysis}
	1) The D2D transmission energy is relatively small and can be ignored; %\footnote{\rev{In terms of the actual D2D transmission energy, paper \cite{exp2}  showed that the energy consumption of non-cooperation and cooperation are similar, which implies that the D2D transmission energy is not significant.}} 
	2) there is no delay constraint; 3) devices can only cooperate with their one-hop neighbors.
\end{assumption}

Under Assumption \ref{ass:analysis}, for any device $m$, the optimal allocation of any of its tasks $s\in\boldsymbol{s}_m$ is as follows. Regarding the inputting subtask, for a content requested by task $s$, if any device $n\in\mathcal{E}(m)$ has cached it, then the content will be inputted from device $n$. If none of the devices in set $\mathcal{E}(m)$ has cached it, then the device  who has the highest downloading capacity among set $\mathcal{E}(m)$  downloads the content. Regarding the computation and uploading subtasks, they will be allocated to the devices with the highest computation and uploading capacities among the devices  $\mathcal{E}(m)$, respectively. 

\subsection{Energy Reduction Due to  the 3C Framework}
In this subsection, we study how much a 3C framework can reduce the energy consumption through providing more cooperation opportunities. 

We explain the basic analysis idea using the following simple example. Suppose among the entire device population, $\alpha N$ devices initialize downloading tasks and participate in user-provided network, and another $\alpha N$ devices initialize computation tasks and participate in ad hoc computation offloading, where $\alpha$ refers to a fraction of devices, and we assume that these two set of devices do not overlap with each other. Under  1C models (e.g., \cite{upn1} and \cite{cloudlet1}), only devices with the same type of tasks (hence requesting the same type of resource) can cooperate; in our proposed 3C framework, all $2 \alpha N$ devices can cooperate with each other, so that the number of devices sharing each of the downloading and computation resources is  doubled, respectively. We are interested in calculating the energy gap between these two cooperative scenarios, showing the energy reduction due to the 3C framework. To clarify, the above scenario is only a simple example, our analysis will cover not only the communication and computation resources but also the caching resource.

Since that each of the tasks only requests one kinds of the 3C resources, we can analyze the tasks requesting each of the 3C resources separately. Specifically, as in the above example, we can calculate the energy reduction of the tasks requesting downloading resource and the tasks requesting computation resource separately, and the entire energy reduction will be the sum of the two energy reductions. 

Next we will compute the  energy reduction of the tasks requesting each of the 3C resources one by one. We will first discuss the tasks requesting communication/computation resource (both of which are capacity-based resources), and then discuss the tasks requesting caching resource.

\subsubsection{\textbf{Communication/Computation}}\label{subsec:size-1}
In the following analysis, we focus on the tasks requesting a particular  resource (i.e., downloading, uploading, or computation). Hence, for presentation simplicity, the term ``resource" in Section \ref{subsec:size-1} only refers to the particular resource, and we omit the corresponding resource-specific super-scripts and sub-scripts. Without the loss of generality, we normalize the energy coefficient of the particular resource to be one, i.e., $c = 1$.

We will first formulate  the expected energy consumption, then introduce the analysis idea. In the random graph $G(N,p)$, suppose each device joins the cooperative system with a probability $\alpha\in[0,1]$. Under these, each task requesting the particular resource will have an expected energy consumption denoted by $W(\alpha,Np)$.\footnote{Under the homogeneous distribution settings in Section \ref{subsec:analysis-setting}, all the tasks will have the same expected energy consumption, so we only need to study the expected energy consumption of a task.}  Under the 1C model, let us denote the probability that  each device joins the corresponding 1C model (that shares the particular resource) as $\alpha^{1C}\in[0,1]$; under the 3C framework, the corresponding probability  is $\alpha^{3C}=\min\{r\alpha^{1C},1\}$, where $r\geq 1$ is a coefficient reflecting the ratio of the increased cooperation opportunities. We will compute the energy reduction $\Delta W(r,\alpha^{1C},Np) \triangleq W(\alpha^{1C},Np)-W(\alpha^{3C},Np)$ for $r\geq1$.

First, we calculate the expected energy, i.e., $W(\alpha,Np)$, under particular $\alpha$ and $Np$. As we have explained, under Assumption \ref{ass:analysis}, any device's task will be allocated to its neighbor who has the highest capacity. By using the order statistic result \cite{orderstatistics}, the probability density function of the highest capacity among a total of $n$ devices is given by
\begin{equation}
f_{(n)}	(x) =  n(F(x))^{n-1}f(x) .
\end{equation}
In the random graph, for a device with a degree $m$, the probability that $\hat{N}$ of its neighbors  join in the cooperative system  is $P(\hat{N}|m) = C_{m}^{\hat{N}}\alpha^{\hat{N}}(1-\alpha)^{m-\hat{N}}$, and the corresponding distribution of the highest capacity among these $\hat{N}$ devices and itself is $f_{(\hat{N}+1)}(x)$. Taking the expectation over $\hat{N}=\{0,...,m\}$, the expected energy consumption of this device's task is given by
\begin{equation}\label{eq:order}
\hat{W}_m(\alpha,Np)= \sum_{\hat{N}=0}^{m} P(\hat{N}|m) \int_{\underline{Q}}^{\overline{Q}} \frac{1}{x} 	f_{(\hat{N}+1)}(x) dx.
\end{equation}
Taking the expectation of $\hat{W}_m(\alpha,Np)$ over all degrees $m=\{0,...,\infty\}$\cite{networkscience}, {the expected energy} of a task is 
\begin{equation}\label{eq:CPU-N-lower}
{W(\alpha,Np)} = \frac{1}{\overline{Q}}+ \int^{\overline{Q}}_{\underline{Q}}e^{ N p(F(x)-1)\alpha}F(x)x^{-2}dx,
\end{equation}
with the detailed proof  given in Appendix \ref{app:eq27}.

Then, we discuss how much the 3C framework can reduce the energy consumption under a coefficient $r\geq 1$. We are interested in the best (the maximum energy reduction) that  the 3C framework can achieve for any $\alpha^{1C}$ and $p$ under an $r$, i.e., $\max_{\alpha^{1C},p} \Delta W(r,\alpha^{1C},Np)$. %\footnote{It is less interesting to discuss the worst case (the minimum energy reduction). This is because under the worst case, the energy reduction is always zero, i.e., $\min_{\alpha^{1C},p}\Delta W(r,\alpha^{1C},Np)  = 0$, which is achieved when either noncooperation or full cooperation occurs in both the 1C and 3C approaches.} %which is achieved under $p=0$, $p=1$, $\alpha^{1C}=0$, or $\alpha^{1C}=1$ (where either noncooperation or full cooperation occurs in both the 1C and 3C approaches).}
The maximum energy reduction that is caused by the 3C framework is as follows:
\begin{theorem}[Maximum Energy Reduction of Communication/Computation]\label{thm:maxenergy}
	Under a coefficient $r\geq 1$, the maximum energy reduction due to the 3C framework is given by
	\begin{multline}
	\max_{\alpha^{1C},p} \Delta W(r,{\alpha^{1C},p}) %\triangleq \max_{\alpha^{1C},p} \left(W(\alpha^{1C},Np)-W(\alpha^{3C},Np)\right)
	\\
	= \int^{\overline{Q}}_{\underline{Q}}\left(e^{\frac{N\tilde{p}(F(x)-1)}{r}}-e^{N\tilde{p}(F(x)-1)}\right)F(x)x^{-2}dx,
	\end{multline}
	where $\tilde{p}$ satisfies 
	\begin{equation}
	\int^{\overline{Q}}_{\underline{Q}}(F(x)-1)F(x)\left(e^{\frac{N\tilde{p}(F(x)-1)}{r}}- re^{N\tilde{p}(F(x)-1)}\right)dx = 0.
	\end{equation}
\end{theorem}
The proof is given in Appendix \ref{app:thm-maxenergy}. \rev{The key idea is to show that the non-concave energy reduction $\Delta W(r,\alpha^{1C},Np)$ has a unique maximizer, which satisfies the first order condition.} 

\rev{Theorem \ref{thm:maxenergy} shows the maximum energy reduction under a general capacity distribution $F(x)$. In order to reveal practical insights, we show a concrete example.}
\begin{figure}[t]
	\centering
	\includegraphics[height=2.5cm]{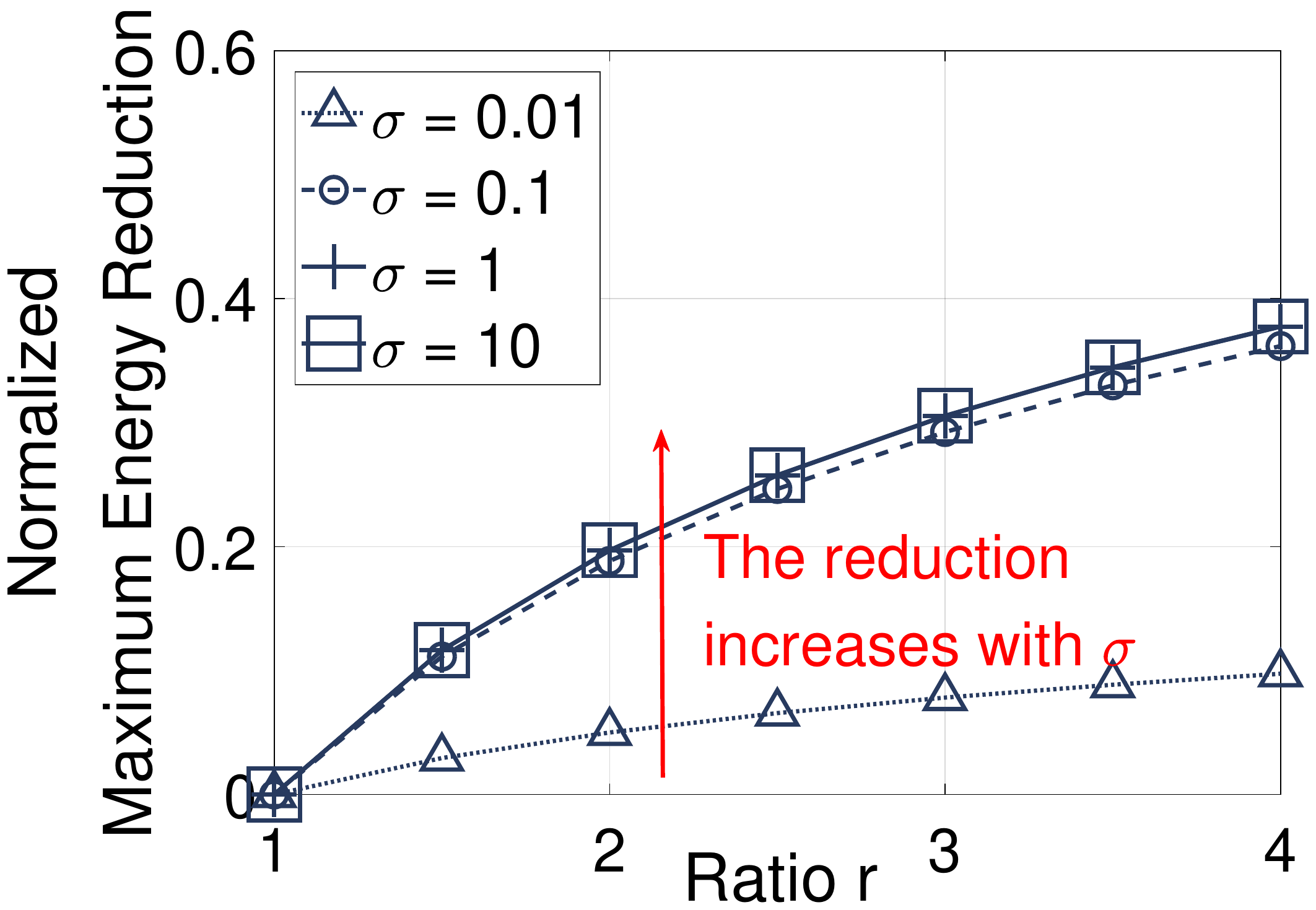}
	\includegraphics[height=2.5cm]{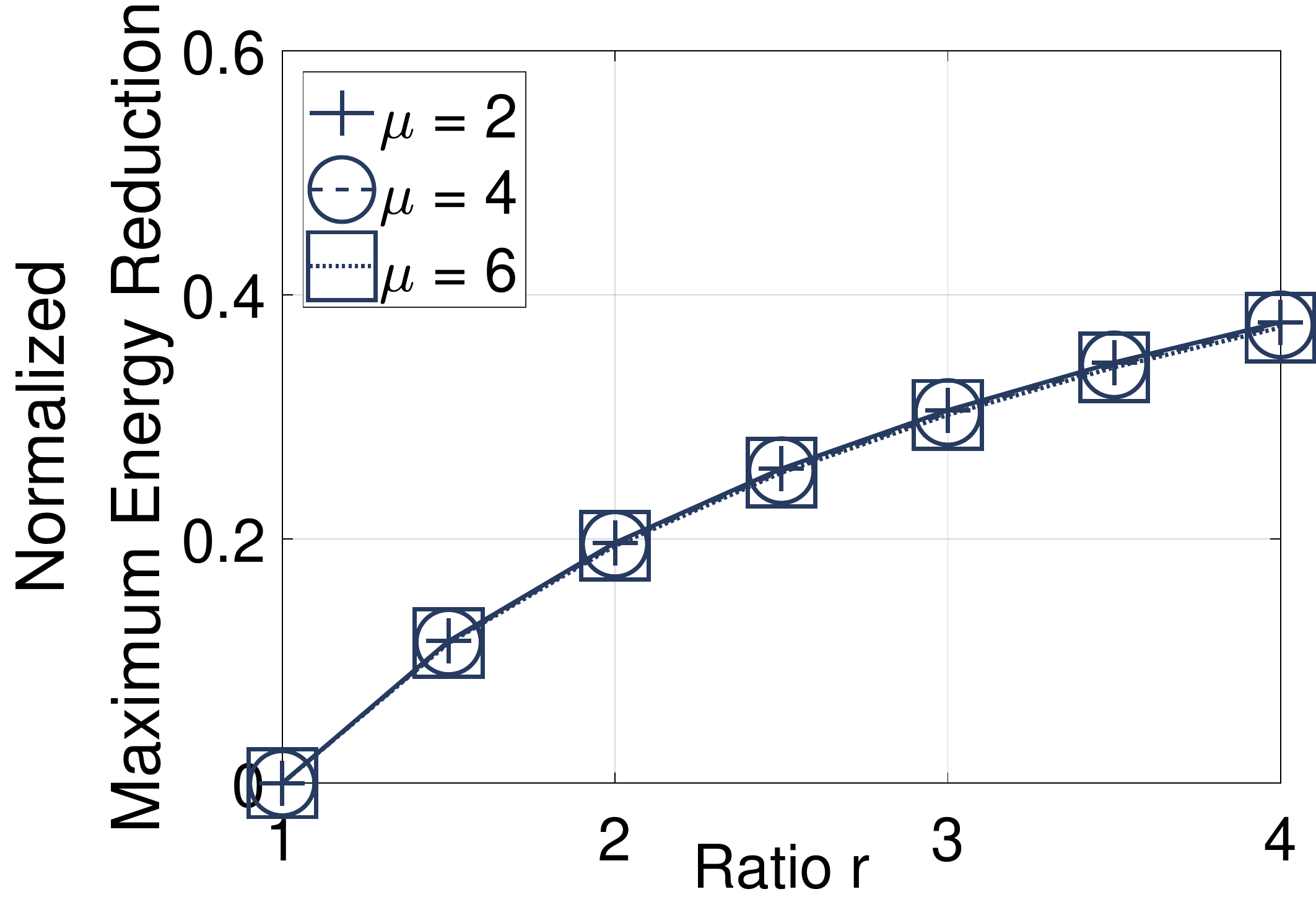}\\
	~~~~~(a)~~~~~~~~~~~~~~~~~~~~~~~~~~~~~~(b)\vspace{-2mm}\\
	\caption{Maximum energy reduction normalized by the energy consumption in noncooperation: (a) $\mu=2$; (b) $\sigma=10$.}\label{fig:example}\vspace{-4mm}
\end{figure}
\begin{example}\label{exampe:1}
	Let us consider the truncated normal distribution $F(x) = F(x;\mu,\sigma,a,b)$, which can be regarded as a normal distribution $N(\mu ,\sigma^2)$ that lies within the interval $[a,b]$ (please refer to  Appendix \ref{app:trun} and \cite{distribution} for details). Figure \ref{fig:example} shows the maximum energy reduction (normalized by the energy consumption in the  noncooperation case,  $W(0,0)$). From Figure \ref{fig:example}, we conclude as follows. \rev{(i) The energy reduction is higher when the variance $\sigma$ is larger. Intuitively, when the devices and tasks are more heterogeneous, the framework benefits more from exploiting the devices' and tasks' heterogeneities.} \rev{(ii) The energy reduction is not affected by the mean $\mu$, i.e., the average capacity of the devices.}
	(iii) Under a large variance $\sigma$ (e.g., $\sigma = 10$), doubling the  sharing devices fraction (i.e., $r=2$) leads to a maximum energy reduction of around $20\%$ of the energy consumed in noncooperation.
\end{example}

\subsubsection{\textbf{Caching}}
The analysis for caching is similar as that  for the communication/computation, where the details are given in  Appendix \ref{app:eqcache}. The expected energy reduction of a content  ${Z}(\alpha,Np)$ is as follows:\footnote{Similarly, under the homogeneous distribution settings in Section \ref{subsec:analysis-setting}, each content will have the same expected energy consumption, so we only need to study the expected energy consumption of a content.}
\begin{equation}\label{eq:cache}
{Z}(\alpha,Np)  = (1-\frac{\Mcache}{K})e^{-\alpha Np \frac{\Mcache}{K}}.
\end{equation}
Under this, the energy reduction due to 3C framework is  $\Delta Z(r,\alpha^{1C},Np) \triangleq Z(\alpha^{1C},Np)-Z(\alpha^{3C},Np)$. Then, the maximum energy reduction is given as follows.
\begin{theorem}[Maximum Energy Reduction of Caching]\label{thm:2}
	Under a coefficient $r\geq 1$, the maximum energy reduction due to the 3C framework is given by
	\begin{multline}\label{eq:reduction-caching}
	\max_{\alpha^{1C},p} \Delta Z(r,\alpha^{1C},Np)
	\\
	= \left(1-\frac{M^{ca}}{K}\right)\left(e^{-\frac{\ln r}{(r-1)}} - e^{-\frac{r\ln r}{ (r-1)}}\right).
	\end{multline}
\end{theorem}
The proof is given in Appendix \ref{app:thrm2}. \rev{The idea is to show that the energy reduction $\Delta Z(r,\alpha^{1C},Np)$  has a unique maximizer, which induces the maximum energy reduction in \eqref{eq:reduction-caching}.}

\rev{Theorem \ref{thm:2} shows the maximum energy reduction due to  the 3C framework in terms of cached content sharing.} For a better understanding of \eqref{eq:reduction-caching}, we normalize such a maximum energy reduction with respect to the energy consumption in the noncooperation case (i.e., $Z(0,0)$). The normalized energy reduction is given by $e^{-{\ln r}/{(r-1)}} - e^{-{r\ln r}/{ (r-1)}}$, which has a similar increasing concave shape as the curves in the two subfigures in Figure \ref{fig:example}. Based on the normalized energy reduction, we conclude as follows. (i) The  normalized maximum energy reduction is independent of the caching ratio ${M^{ca}}/{K}>0$. This means that  no matter how many contents that devices have cached, the maximum normalized energy reduction is fixed. (ii) Doubling the sharing device fraction (i.e., $r=2$) leads to a maximum energy reduction of around $25\%$ of the energy consumed in noncooperation.

\section{Simulation and Performance}\label{sec:experiment}

We compare the computation time and energy consumption between optimal and heuristic solutions. And we evaluate the energy reduction due to 3C framework under different D2D transmission energy and different devices' and tasks' heterogeneities. To emphasize, these simulations are based on a more realistic case, where Assumption \ref{ass:analysis} is relaxed.

We consider a scenario with a set of $N$ devices, who form pair-wise connections with a probability $p=0.3$. Each device has one task to execute. For each simulation setting, we perform 100 rounds  and show the average results. For each simulation round, we randomly generate the parameters of the device and task models, including devices' capacities, tasks' demands, and energy consumption coefficients. These parameters are randomly generated based on truncated normal distributions\cite{distribution}, with an identical variance $\sigma$ (which will be evaluated later) and different means. The detailed settings are in Appendix \ref{app:exp}. 

\vspace{-3mm}

\subsection{Comparison: Optimal and Heuristic Solutions}\label{subsec:compare}
We show how the device number $N$ affects the computation time and energy consumption of the optimal (named as ``Opt.") and the heuristic algorithm (named as ``Heu.").

Figure \ref{fig:scale} (a) shows how the total computation time changes in $N$.%\footnote{\rev{The actual runtime of the ``OPT" and  ``Heu" highly depend on the computation speed of the device. In the case that the actual runtime is not small enough for the ultra-responsive requirement of the Tactile Internet, the ``OPT" and  ``Heu" may be applied to problems whose resource allocation decisions do not need to be made frequently.}} 
For the case of ``Opt." (solving Problem $\text{(OPT-LINEAR)}$), the  computation time is small when $N$ is small (e.g., $N$ is less than $20$). However,  as the device number increases, the computation time of ``Opt." dramatically increases. In comparison,  the computation time of ``Heu." increases relatively slowly in $N$, i.e., the computation time of ``Heu." is $78.6\%$ smaller than that of ``Opt." when $N=27$.%\footnote{When $N$ is smaller than $20$, the computation time of ``Opt." is smaller than that of ``Heu.". This is because executing ``Opt." requires solving one ILP problem, while executing ``Heu." requires iteratively solving multiple LP problems. Hence, implementing ``Heu." is beneficial only when $N$ is large, under which the NP-complete problem $\text{(OPT-LINEAR)}$ (of executing ``Opt.") induces extreme large computation time.}

Figure \ref{fig:scale} (b) shows the energy comparison between ``Opt." and ``Heu.". The energy is normalized by the energy consumed in the noncooperation  case (i.e., each device executes its task by itself). %\footnote{We skip the simulation of $N=30$ due to its huge computation time.} 
As $N$ increases, the energy gap between ``Opt." and ``Heu." slightly increases. When $N=27$, the normalized percentage difference of the energy between ``Heu." and ``Opt." is only around $11.2\%$.

\begin{figure}
	\centering
	\includegraphics[height=2.8cm]{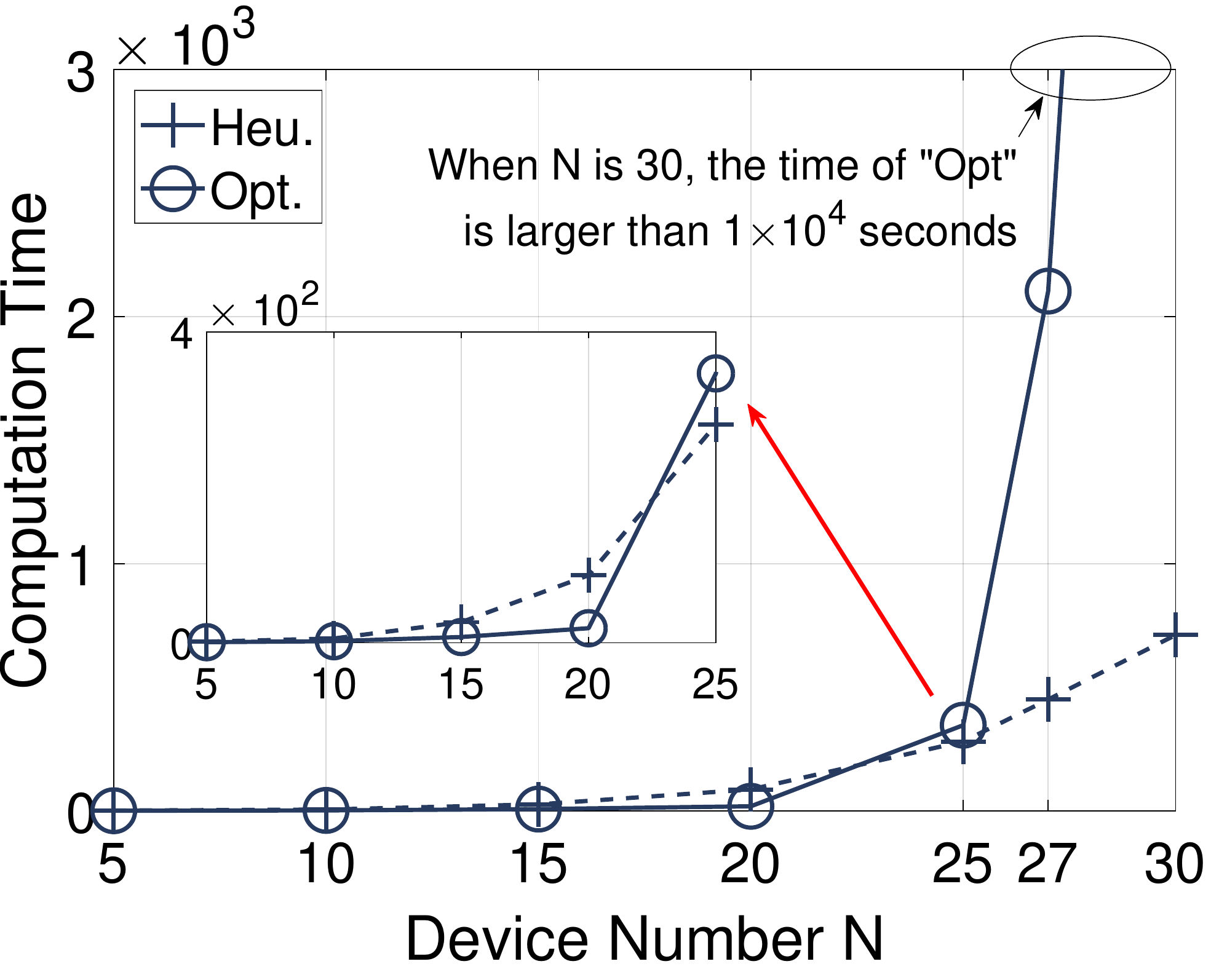}~~
	\includegraphics[height=2.8cm]{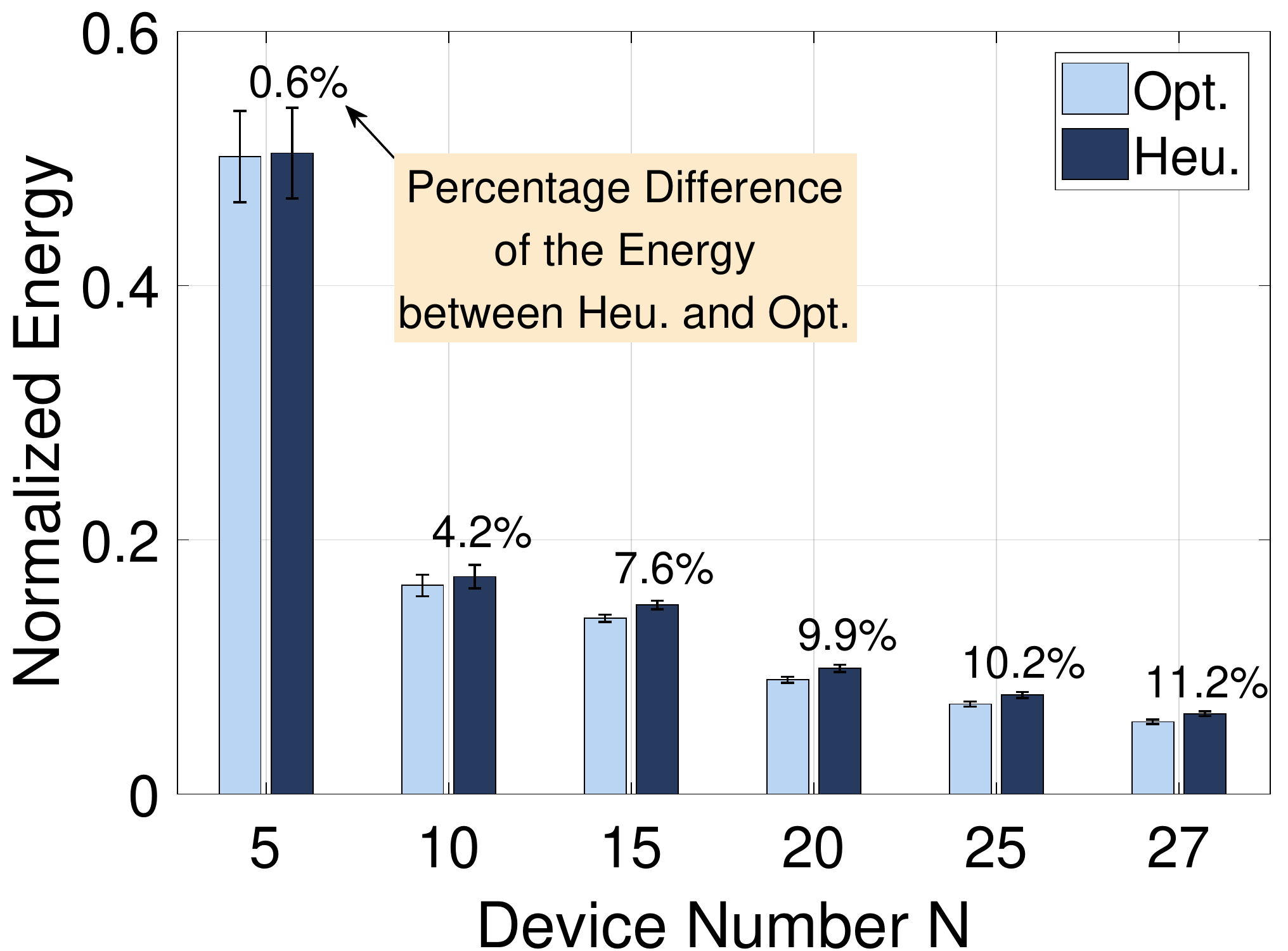}\\
	~~~~~~(a)~~~~~~~~~~~~~~~~~~~~~~~~~~~~~~~~~(b)~~\vspace{-2mm}\\
	\caption{Impact of $N$ on (a) computation time and  (b)  energy consumption.}\label{fig:scale}\vspace{-3mm}
	
\end{figure}

\vspace{-3mm}
\subsection{Comparison: 1C/2C Models and 3C Framework}
In this simulation, we let each device randomly selects a task among downloading, content sharing, and distributed data analysis. Then, we perform simulations in two cooperation settings: (i) ``1C/2C", where only the devices selecting the same kinds of tasks can cooperate; (ii) ``3C", where all the devices can cooperate.

In Figure \ref{fig:1c2c3c}, we compare the energy of the two cooperation settings under different D2D energy coefficients $\beta$ (the ratio of the D2D energy per unit time  to the downloading energy per unit time) and variances $\sigma$ (the variance for generating tasks and devices). \rev{Note that the energy consumption is normalized by the energy consumed in the noncooperation case, e.g., a normalized energy value 0.2 means that the cooperation (1C/2C or 3C) consumes 20\% of the energy consumed in the noncooperation.} The percentage reduction in the figure is the energy difference between ``1C/2C" and ``3C", normalized by the energy consumed in ``1C/2C".

In Figure \ref{fig:1c2c3c} (a), ``3C" can reduce the energy consumption by 83.8\%  when $\beta=0$, i.e., no additional energy consumption due to D2D transmission. Such an energy reduction decreases in $\beta$, but still achieves a value of 27.5\% when $\beta=2.0$, i.e., the  D2D energy per unit time is twice as large as the downloading energy per unit time.

In Figure \ref{fig:1c2c3c} (b), as the heterogeneity of devices and tasks (measured by the variance $\sigma$) increases, the energy reduction caused by 3C framework increases, which is consistent with the results in Example \ref{exampe:1} in Section \ref{sec:theory}. Intuitively, a higher heterogeneity can provide more opportunities for the devices to share resources and help each other. Hence, implementing the 3C framework is more beneficial when devices and tasks are more heterogeneous.

\begin{figure}
	\centering
	\includegraphics[height=2.8cm]{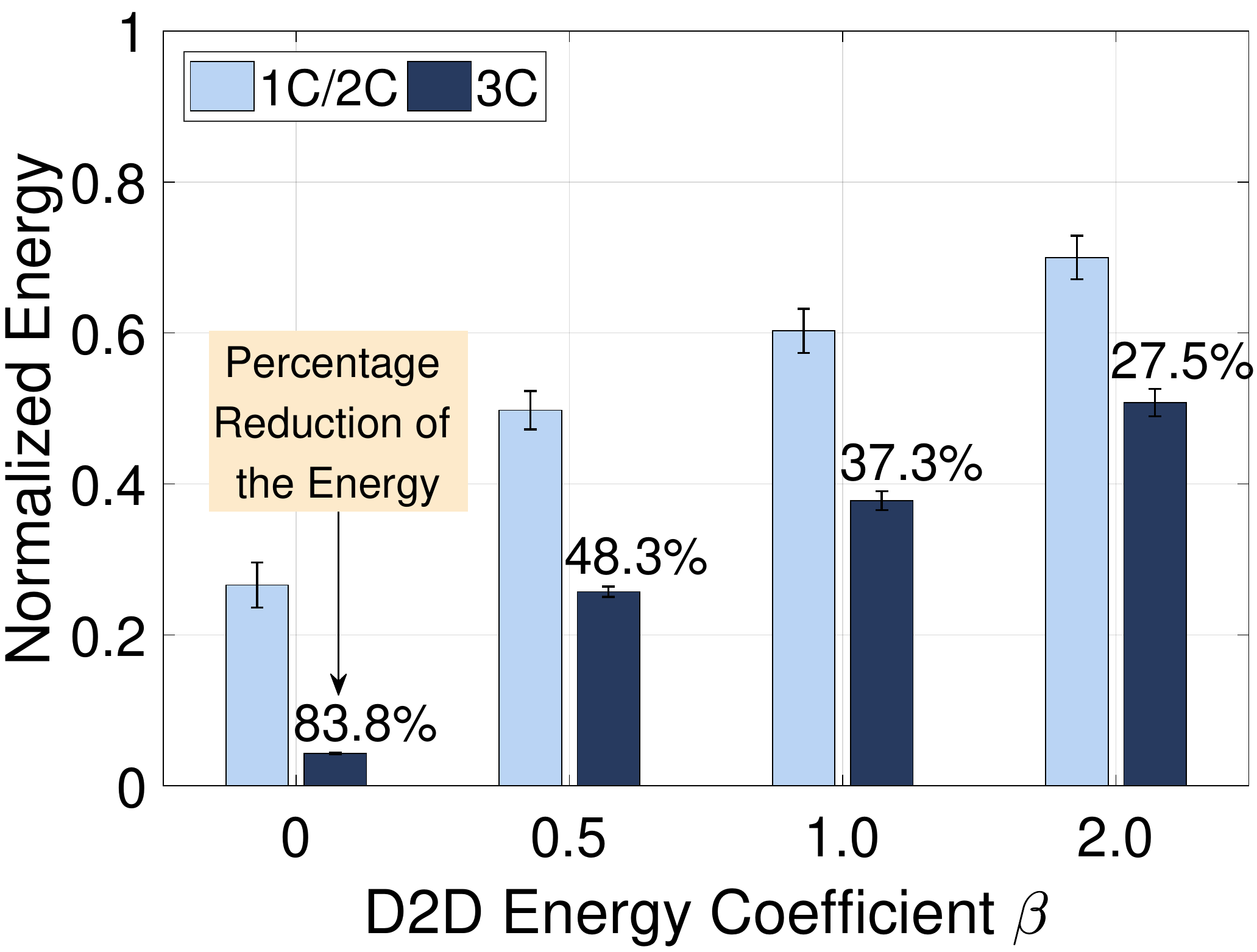}~~\includegraphics[height=2.8cm]{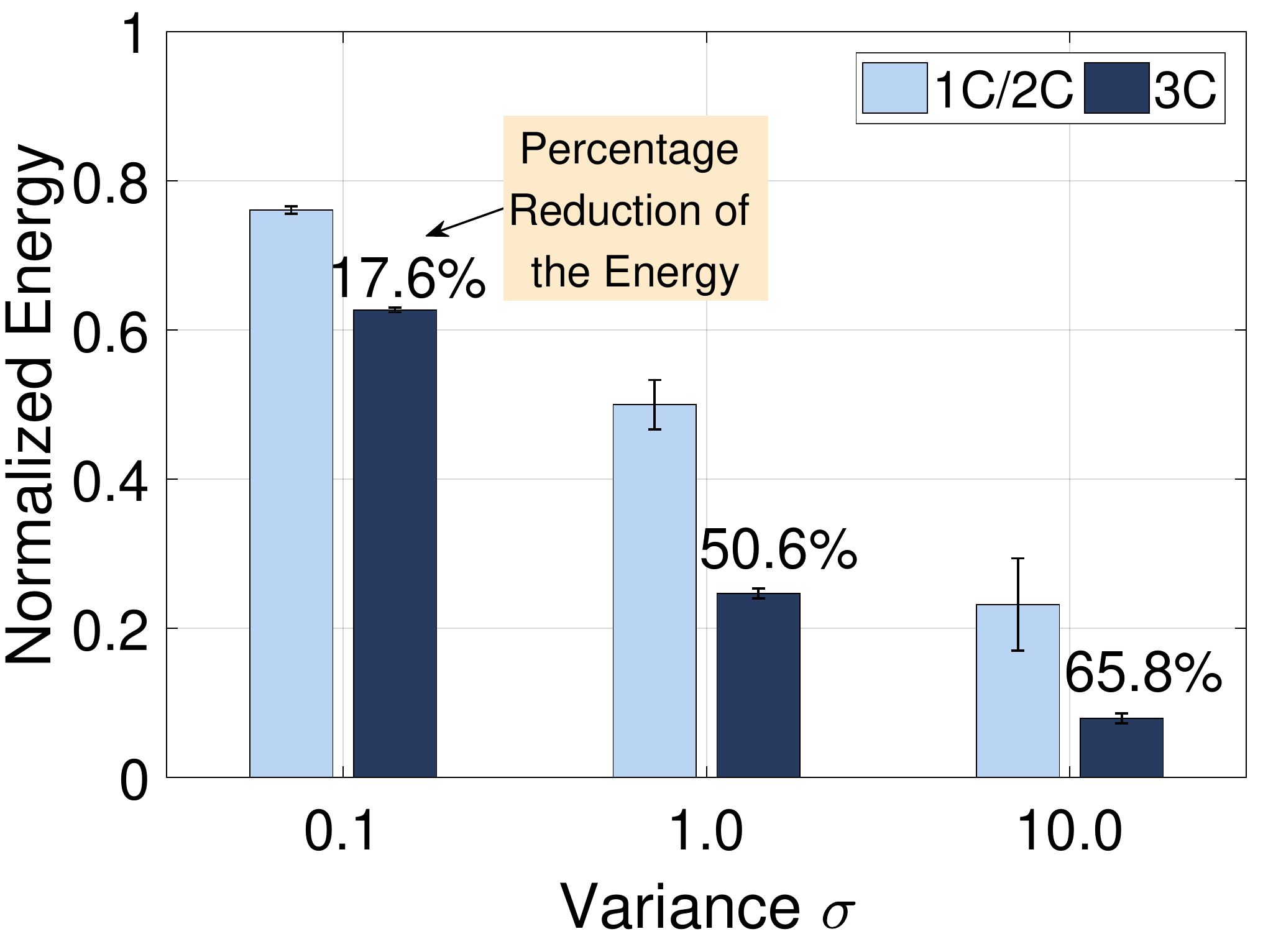}\\
	~~~~~~(a)~~~~~~~~~~~~~~~~~~~~~~~~~~~~~~~~~(b)~~\vspace{-2mm}\\
	\caption{Impact of D2D energy $\beta$ (under $\sigma = 1.0$) and variance $\sigma$ (under $\beta=0.5$) on the  normalized energy consumption.}\label{fig:1c2c3c}\vspace{-3mm}
\end{figure}

\section{Conclusion}\label{sec:conclude}
In this paper, we propose a general 3C framework that enables the joint 3C resource sharing among mobile edge devices, which potentially enhances the reliability and intelligence of Tactile Internet. This  ``resource-centric" framework generalizes existing  D2D resource sharing models, and provides a structure for future D2D resource sharing analysis. We theoretically and numerically show that the 3C framework can further exploit resource sharing potentials and improve resource utilization efficiency significantly. For future work,  \rev{it is interesting to design a distributed algorithm for the 3C framework, where the information collection and allocation scheduling are operated in a distributed fashion. %a stochastic model that tasks and devices arrive sequentially, where we need to optimize the resource sharing with uncertain  future information.
	It is also interesting to consider the optimization of proactive caching, where a user may cache contents for the use of a future task.}

\vspace{-10mm}
\begin{IEEEbiography}
	[{\includegraphics[width=1in,height=1.25in,clip,keepaspectratio]{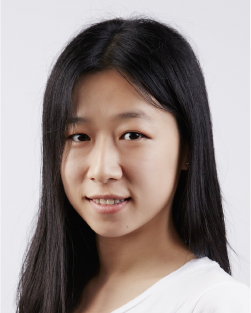}}]{Ming Tang}
	(S'16) is currently pursuing a PhD degree at the Department of Information Engineering, The Chinese University of Hong Kong. She was a visiting student with the Department of Management Science and Engineering, Stanford University, from September 2017 to February 2018. Her research interests include wireless communications and network economics, with particular emphasis on user-provided networks and fog computing. %She is a student member of the IEEE.
\end{IEEEbiography}
\vspace*{-2.3\baselineskip}
\begin{IEEEbiography}
	[{\includegraphics[width=1in,height=1.25in,clip,keepaspectratio]{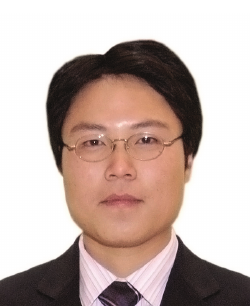}}]{Lin Gao}
	(S'08-M'10-SM'16) received the PhD degree in electronic engineering from Shanghai Jiao Tong Unviersity, in 2010. He is an Associate Professor with the School of Electronic and Information Engineering, Harbin Institute of Technology, Shenzhen, China. His main research interests are in the area of network economics and games, with applications in wireless communications and networking. He received the IEEE ComSoc Asia-Pacific Outstanding Young Researcher Award in 2016.
\end{IEEEbiography}
\vspace*{-2.3\baselineskip}
\begin{IEEEbiography}
	[{\includegraphics[width=1in,height=1.25in,clip,keepaspectratio]{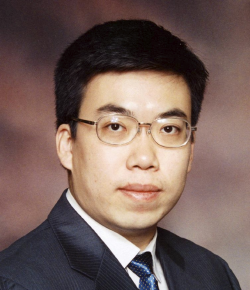}}]{Jianwei Huang}
(F'16)  is a Professor in the Department of Information Engineering at the Chinese University of Hong Kong. He is the co-author of 9 Best Paper Awards, including IEEE Marconi Prize Paper Award in Wireless Communications 2011. He has co-authored six books, including the textbook on "Wireless Network Pricing”. He has served as the Chair of IEEE Technical Committee on Cognitive Networks and Technical Committee on Multimedia Communications. He is an IEEE ComSoc Distinguished Lecturer and a Clarivate Analytics Highly Cited Researcher. More detailed information can be found at \url{http://jianwei.ie.cuhk.edu.hk/}.
\end{IEEEbiography}

\appendix

\section{Appendix}
\subsection{Proof for Proposition \ref{prp:1}}\label{app:prp-1}
We first present a lemma showing that all the extreme points of the feasible region in Problem $\text{(OPT-RELAX)}$ are integers. Then, we present a lemma showing that solving Problem $\text{(OPT-RELAX)}$ using Simplex method\cite{simplex} is guaranteed to  output an integer solution. Finally, we explain that Algorithm \ref{alg:suboptimal} is guaranteed to produce an integer solution that is within the feasible region of Problem $\text{(OPT)}$.

\begin{lemma}[Integer Extreme Points]\label{lem:integer}
	%The feasible region of optimization problem $\text{(OPT-RELAX)}$ is bounded, and 
	All the extreme points of the feasible region in Problem $\text{(OPT-RELAX)}$ are integers.
\end{lemma}
\begin{proof}
	To prove this lemma, we first transform Problem $\text{(OPT-RELAX)}$ into an equivalent problem. Then, we prove that the equivalent problem has only integer extreme points.
	
	First, we substitute constraints $\eqref{constr:flow-in}\sim\eqref{constr:flow-up}$ into the objective function, leading to an equivalent optimization problem in the following form:
	\begin{equation}\label{eq:problem4c}
		\begin{aligned}
			& {\text{minimize}} & & \boldsymbol{h}_1^{\mathrm{T}}\boldsymbol{x} + \boldsymbol{h}_2^{\mathrm{T}}{\boldsymbol{z}}\\
			& \text{subject to} & &  \boldsymbol{a}\leq \left[\boldsymbol{A}~\boldsymbol{0}\right]\left[\begin{array}{c}\boldsymbol{x}\\{\boldsymbol{z}}\end{array}\right] \leq  \boldsymbol{b}\\
			&&& \boldsymbol{c}\leq \left[\begin{array}{c}\boldsymbol{x}\\{\boldsymbol{z}}\end{array}\right] \leq \boldsymbol{d}\\
			& \text{variables}
			& & \boldsymbol{x},{\boldsymbol{z}}\in [0,1] \\
		\end{aligned}
	\end{equation}
	The first constraint corresponds to  $\eqref{constr:allo-cpu}\sim\eqref{constr:cap-in}$, and the second constraint corresponds to $\eqref{constr:control-in}$. %Variables $\hat{\boldsymbol{z}}$ is a subset of the variables $\boldsymbol{z}$, which is derived based on equalities $\eqref{constr:flow-in}\sim\eqref{constr:flow-up}$. %and the optimal solution of the rest set $\boldsymbol{z}/\hat{\boldsymbol{z}}$ can be obtained through equalities $\eqref{constr:flow-in}\sim\eqref{constr:flow-up}$.
	
	Then, we show that the problem \eqref{eq:problem4c} has only integral extreme points (so as Problem  $\text{(OPT-RELAX)}$). The idea is to show that  $\boldsymbol{a}, \boldsymbol{b},\boldsymbol{c},\boldsymbol{d}$ contain only integers and $\left[\boldsymbol{A}~\boldsymbol{0}\right]$ is totally unimodular \cite{unimodular}. Through checking the constraints, we directly have $\boldsymbol{a}, \boldsymbol{b},\boldsymbol{c},\boldsymbol{d}$ contain only integers. According to \cite{unimodular2}, matrix $\left[\boldsymbol{A}~\boldsymbol{0}\right]$ is totally unimodular, because  (i) all the elements belong to $\{+1,-1,0\}$, (ii) each column of the matrix contains at most two non-zero elements (only $x^{in}_{s,k\rightarrow n}$ has two non-zero elements from constraints \eqref{constr:allo-in} and \eqref{constr:cap-in}, while the others only have at most one), 
	(iii) for the only variables with two non-zero elements $x^{in}_{s,k\rightarrow n}$ from constraints \eqref{constr:allo-in} and \eqref{constr:cap-in}, the two non-zero elements have the same sign and can be separated into two disjoint sets (of rows) by separating constraints \eqref{constr:allo-in} and \eqref{constr:cap-in} into two subsets of $\boldsymbol{x}$ by rows.	
	%	Finally, problem \eqref{eq:problem4} is a linear optimization problem with bounded varaibles and feasible solutions, where the feasibility of problem  \eqref{eq:problem4} can be proved by showing that it is always a feasible point that devices perform tasks on their own. So the optimal solution will occur at a vertex, or on a line segment between two vertices. Since we have proved that the vertices are integers, Lemma \ref{lem:integer} holds.
	%Finally, since the problem \eqref{eq:problem4c} is a linear programming and has integral variables, the solutions to it are all integers, so as the equivalent problem \eqref{eq:problem4}.
	%In problem \eqref{eq:problem4},  the tasks' decisions are independent[NOT, correlated due to the downloading behavior], in the sense that the problem can be separated into $S$ subproblems (each for a task $s\in\S$) but still achieves the same solutions. Hence, we can focus on a subproblem for a task $s$, which minimizes the task $s$'s energy consumption under the constraints for task $s$. 
	%Second, according to constraints $(8)\sim (10)$, if variables $\boldsymbol{x}$ are integers, then variables $\boldsymbol{z}$ should be also integers. %\textcolor{blue}{A brief idea is as follows. Suppose some solutions are not integer, we can always find corresponding integer solutions such that its energy consumption is no larger than that of the original non-integer solutions.}	
\end{proof}

\begin{lemma}[Integer Output of Problem $\text{(OPT-RELAX)}$]\label{lem:integer-output} If solving Problem $\text{(OPT-RELAX)}$ using Simplex method\cite{simplex}, the output solution is an integer solution.
\end{lemma}
\begin{proof}
	Problem $\text{(OPT-RELAX)}$ is an LP problem that has bounded variables and is feasible, where the feasibility  can be proved by showing that noncooperation is always a feasible point of Problem $\text{(OPT-RELAX)}$ (as assumed in Assumption \ref{ass:noncooperation}). So there always exists a vertex that is the optimal solution of Problem $\text{(OPT-RELAX)}$\cite{unimodular}. One the other hand, Simplex method solves LP by traversing the edges between vertexes on the feasible region, such that the output solution is always a vertex. Hence, the output solution (using Simplex method to solve Problem $\text{(OPT-RELAX)}$) is an integer solution.	
	%	 Using Simplex method,... Since we have proved that the vertexes are integers, Lemma \ref{lem:integer} holds.
	%Finally, since the problem \eqref{eq:problem4c} is a linear programming and has integral variables, the solutions to it are all integers, so as the equivalent problem \eqref{eq:problem4}.
	%In problem \eqref{eq:problem4},  the tasks' decisions are independent[NOT, correlated due to the downloading behavior], in the sense that the problem can be separated into $S$ subproblems (each for a task $s\in\S$) but still achieves the same solutions. Hence, we can focus on a subproblem for a task $s$, which minimizes the task $s$'s energy consumption under the constraints for task $s$. 
	%Second, according to constraints $(8)\sim (10)$, if variables $\boldsymbol{x}$ are integers, then variables $\boldsymbol{z}$ should be also integers. %\textcolor{blue}{A brief idea is as follows. Suppose some solutions are not integer, we can always find corresponding integer solutions such that its energy consumption is no larger than that of the original non-integer solutions.}	
\end{proof}

Finally, we explain that Algorithm \ref{alg:suboptimal} is guaranteed to produce an integer solution that is within the feasible region of Problem $\text{(OPT)}$. There will be two claims: first, Algorithm \ref{alg:suboptimal} always has an output; second, if  Algorithm \ref{alg:suboptimal} has an output, the output is a feasible solution of Problem $\text{(OPT)}$. The first claim is directly held under Assumption \ref{ass:noncooperation}. The second claim is true because (i) the output is an integer solution satisfying constraints $\eqref{constr:allo-cpu}\sim\eqref{constr:flow-up}$ (according to Lemma \ref{lem:integer-output}), and (ii) the output satisfies the delay constraint $\eqref{eq:delay}$.

\subsection{Proof for Proposition \ref{prp:2}}\label{app:prp-2}
We prove the two claims in this proposition one by one.

First, the energy consumption of the heuristic algorithm output is no larger than that of the noncooperation case for the following reason. The output of Algorithm \ref{alg:suboptimal} is an optimal solution to Problem $(\text{OPT-RELAX})$ under  a particular set of control parameters $\tilde{\boldsymbol{N}}^{in},\tilde{\boldsymbol{N}}^{cpu},\tilde{\boldsymbol{N}}^{up}$. Under the same set of control parameters, the noncooperation case is also a feasible point of Problem $(\text{OPT-RELAX})$. This implies that the energy consumption of the heuristic algorithm output is no larger than that of the noncooperation case; otherwise, the output cannot be an optimal solution of Problem $(\text{OPT-RELAX})$ under the set of control parameters. 

Second, when there is no delay constraint,  the heuristic algorithm output is an optimal solution of the original problem $\text{(OPT)}$ for the following reason. 
When there is no delay constraint, Algorithm \ref{alg:suboptimal} terminates in the first iteration,  so that the output of Algorithm \ref{alg:suboptimal} is the optimal solution of Problem $(\text{OPT-RELAX})$ under all the control parameters are equal to ones. We refer to such version of Problem $(\text{OPT-RELAX})$  as the basic version of Problem $(\text{OPT-RELAX})$. On the other hand, when there is no delay constraint, comparing with Problem $(\text{OPT})$, the only modified part of the basic version of Problem $(\text{OPT-RELAX})$  is that it relaxes integer variables $\boldsymbol{x},\boldsymbol{z}\in\{0,1\}$ to be continuous ones $\boldsymbol{x},\boldsymbol{z}\in[0,1]$. As we proved in Lemma \ref{lem:integer-output}, if solving Problem $(\text{OPT-RELAX})$ using Simplex method, the optimal solution is an integer solution, i.e., $\boldsymbol{x},\boldsymbol{z}\in\{0,1\}$. This means that Problem $(\text{OPT})$ and the basic version of Problem $(\text{OPT-RELAX})$ are equivalent. As a result, when there is no delay constraint,   the heuristic algorithm output is an optimal solution of the basic version of  of Problem $(\text{OPT-RELAX})$, which is equivalent to the original Problem $\text{(OPT)}$.

\subsection{Proof for Proposition \ref{prp:3}}\label{app:prp-3}
We show that Algorithm \ref{alg:suboptimal} will always terminate within $S\times(N-1)$ iterations, where $S$ is the task number and $N$ is the device number.

We can consider the possible allocation of tasks to devices as a bipartite graph: a set of tasks, a set of devices, and a set of links from the tasks to the devices (a task is connected to a device if the task can be potentially allocated to the device). At the beginning of Algorithm \ref{alg:suboptimal}, the bipartite graph is complete (i.e.,  every task is connected to every device), since no allocation is prevented. In each iteration, some links is removed from the bipartite graph by adjusting the control parameters. The maximum iteration happens when (i) each iteration removes one link, (ii) all the links are removed excluding the links from tasks to their task owners, where the maximum iteration time is $S\times(N-1)$.

\subsection{Proof for Equation \eqref{eq:CPU-N-lower}}\label{app:eq27}
The  expected energy of a task whose owner has a degree $m$, i.e., $\hat{W}_m(\alpha,Np)$, can be simplified as follows:
\begin{equation}\label{eq:order}
	\begin{aligned}
		&\hat{W}_m(\alpha,Np)\\%&= \sum_{\hat{N}=0}^{m} P(\hat{N}|m) \int_{\underline{Q}}^{\overline{Q}} \frac{1}{x} 	f_{(\hat{N}+1)}(x) dx\\
		%	&= \sum_{\hat{N}=0}^{m}  C_{m}^{\hat{N}}\alpha^{\hat{N}}(1-\alpha)^{m-\hat{N}}\int_{\underline{Q}}^{\overline{Q}} \frac{1}{x} 	(\hat{N}+1)(F(x))^{\hat{N}}f(x)  dx\\
		= &\sum_{\hat{N}=0}^{m}  C_{m}^{\hat{N}}\alpha^{\hat{N}}(1-\alpha)^{m-\hat{N}}\int_{\underline{Q}}^{\overline{Q}} \frac{1}{x} d(F(x))^{\hat{N}+1}\\
		&~~~~\text{\emph{(Integration by Parts)}}\\
		= &\sum_{\hat{N}=0}^{m}  C_{m}^{\hat{N}}\alpha^{\hat{N}}(1-\alpha)^{m-\hat{N}}\left(\frac{1}{\overline{Q}} +\int_{\underline{Q}}^{\overline{Q}} \frac{(F(x))^{\hat{N}+1}}{x^{2}} dx\right)\\
		&~~~~\text{\emph{(Distributive Property; Binomial Theorem)}}\\
		= & \frac{1}{\overline{Q}} + \int_{\underline{Q}}^{\overline{Q}} \frac{F(x)}{x^2}(1-\alpha+\alpha F(x))^{\hat{N}} dx.
	\end{aligned}
\end{equation}

%	Substituting $P(\hat{N}|m)$ and $f_{(\hat{N}+1)}(x)$ in to \eqref{eq:order}, the expected lower bound energy consumption of a device with degree $m$ (i.e., $\hat{W}(\alpha,m)$) is formulated as follows:
Taking the expectation of $\hat{W}_m(\alpha,Np)$ over all possible degrees $m=\{0,...,\infty\}$, the  expected energy of each task is given as follows:
\begin{equation}
	\begin{aligned}
		&{W(\alpha,Np)} \\
		=& \sum_{m=0}^{\infty} P(degree=m)\hat{W}_m(\alpha,Np)\\	
		=& \sum_{m=0}^{\infty} \frac{(Np)^me^{-Np}\left(\frac{1}{\overline{Q}} + \int_{\underline{Q}}^{\overline{Q}} \frac{F(x)}{x^2}(1-\alpha+\alpha F(x))^{m} dx\right)}{m!} \\
		&~~~~\text{\emph{(Distributive Property; Taylor Series, $\sum_{m=0}^{\infty}\frac{x^m}{m!}=e^x$)}}\\
		%	&=\frac{1}{\overline{Q}} + \int_{\underline{Q}}^{\overline{Q}} \sum_{m=0}^{\infty} \frac{(Np(1-\alpha+\alpha F(x)))^me^{-Np}}{m!}   \frac{F(x)}{x^2} dx\\
		%	&~~~~\text{\emph{(Taylor Series, $e^x = \sum_{m=0}^{\infty}\frac{x^m}{m!}$)}}\\
		=&\frac{1}{\overline{Q}} + \int_{\underline{Q}}^{\overline{Q}} e^{Np( F(x)-1)\alpha}F(x)x^{-2} dx.
	\end{aligned}
\end{equation}
This completes the proof of Equation \eqref{eq:CPU-N-lower}.

\subsection{Proof for Theorem \ref{thm:maxenergy}}\label{app:thm-maxenergy}
We first present two lemmas indicating the optimal $\alpha^{1C}$ and $p$  that maximizes $\Delta W(r,{\alpha^{1C},p})$ for any $r\geq 1$, then show the maximum energy reduction, i.e., $\max_{\alpha^{1C},Np} \Delta W(r,{\alpha^{1C},Np})$.

\begin{lemma}[Optimal $\alpha^{1C}$ under a Particular $Np$]\label{lemma:optimal-alpha}For any ratio $r\geq 1$, connection probability $p$, and distribution $f(x)$, there exists an $\alpha^{1C} = \alpha^*_{Np}$  that maximizes the energy reduction $W(r,{\alpha^{1C},p})$, i.e.,
	\begin{equation}
		\alpha^*_{Np} = \left\{\begin{array}{ll}
			{1}/{r}, & p\leq \tilde{p}\\
			\tilde{\alpha}_{Np},& p>\tilde{p}
		\end{array}\right.,
	\end{equation}
	where $\tilde{\alpha}_{Np}$ satisfies
	\begin{equation}\label{eq:theorem1-alpha0}
		[W(\tilde{\alpha}_{Np},Np)-W(r\tilde{\alpha}_{Np},Np)]_{\alpha} = 0.%\int^{\overline{Q}}_{\underline{Q}}\left(e^{ N p(F(x)-1)\tilde{\alpha}}-e^{ N p(F(x)-1)r\tilde{\alpha}}r\right)(F(x)-1)F(x)x^{-2}dx = 0,
	\end{equation}
	and  connection probability threshold $ \tilde{p}$  satisfies
	\begin{equation}
		\int^{\overline{Q}}_{\underline{Q}}(F(x)-1)F(x)\left(e^{ \frac{N \tilde{p}(F(x)-1)}{r}}-re^{ N \tilde{p}(F(x)-1)}\right) dx = 0.
	\end{equation}
\end{lemma}
\begin{proof}
	The proof path is as follows. We first prove two claims. \emph{Claim (1):} for any $Np$, there exists a unique $\tilde{\alpha}_{Np}$ that maximizes ${W({\alpha},Np)} - {W(r{\alpha},Np)}$, and it satisfies
	\begin{equation}\label{eq:optimal-alpha}
		[{W(\alpha,Np)} - {W(r{\alpha},Np)}]_{\alpha}\left\{\begin{array}{ll}
			> 0,&   \alpha<\tilde{\alpha}_{Np}\\
			= 0,&   \alpha=\tilde{\alpha}_{Np}\\
			< 0,& \alpha>\tilde{\alpha}_{Np}\end{array}\right..
	\end{equation}
	\emph{Claim (2):} there exists a unique $\tilde{p}$ that satisfies
	\begin{equation}\label{eq:optimal-p}
		[{W(1/r, N\tilde{p})} - {W(1, N\tilde{p})}]_{\alpha}\left\{\begin{array}{ll}
			>0,& p<\tilde{p}\\
			=0,& p=\tilde{p}\\
			<0,& p>\tilde{p}\\
		\end{array}
		\right..
	\end{equation}
	Then, using these two claims, we then prove Lemma \ref{lemma:optimal-alpha}.
	
	\emph{Claim (1):} %we prove that there exists a unique $\tilde{\alpha}$ such that $[{W(\tilde{\alpha},Np)} - {W(r\tilde{\alpha},Np)}]_{\alpha}=0$ (as in formulation \eqref{eq:theorem1-alpha0}), where  $[{W(\alpha,Np)} - {W(r{\alpha},Np)}]_{\alpha}> 0$ if  $\alpha<\tilde{\alpha}$, and  $[{W(\alpha,Np)} - {W(r{\alpha},Np)}]_{\alpha}< 0$ if  $\alpha>\tilde{\alpha}$. 
	We first prove that for any $Np$ there exists an $\tilde{\alpha}_{Np}$ that is a local maximizer of ${W({\alpha},Np)} - {W(r{\alpha},Np)}$. Then, we show that such $\tilde{\alpha}_{Np}$  is the global maximizer and satisfies \eqref{eq:optimal-alpha}.
	
	First, we prove that  there exists an $\tilde{\alpha}_{Np}$  that is a local maximizer of ${W({\alpha},Np)} - {W(r{\alpha},Np)}$. Taking the first-order derivative of ${W(\tilde{\alpha},Np)} - {W(r\tilde{\alpha},Np)}$ with respect to $\alpha$,
	\begin{equation}
		\lim_{\alpha\rightarrow 0} [{W(\alpha,Np)} - {W(r\alpha,Np)}]_{\alpha} >0,
	\end{equation}
	\begin{equation}
		\lim_{\alpha\rightarrow \infty}[{W(\alpha)} - {W(r\alpha)}]_{\alpha} \rightarrow 0^-.
	\end{equation}
	According to Intermediate Value Theorem, there exists at least an $\tilde{\alpha}_{Np}$ such that 
	\begin{multline}\label{eq:theorem1-alpha}
		[{W(\tilde{\alpha}_{Np},Np)} - {W(r\tilde{\alpha}_{Np},Np)}]_{\alpha}\\=\int^{\overline{Q}}_{\underline{Q}}\left(e^{ N p(F(x)-1)\tilde{\alpha}_{Np}}-re^{ N p(F(x)-1)r\tilde{\alpha}_{Np}}\right)\\
		\times Np(F(x)-1)F(x)x^{-2} dx=0,\end{multline} 
	which is the $\tilde{\alpha}_{Np}$ that satisfies \eqref{eq:theorem1-alpha0}. In addition, there exist $\Delta_1$ and $\Delta_2$ such that 
	\begin{equation}[{W(\tilde{\alpha}_{Np},Np)} - {W(r\tilde{\alpha}_{Np},Np)}]_{\alpha}>0, \alpha\in(\tilde{\alpha}_{Np}-\Delta_1,\tilde{\alpha}_{Np}),\end{equation}
	\begin{equation}[{W(\tilde{\alpha}_{Np},Np)} - {W(r\tilde{\alpha}_{Np},Np)}]_{\alpha}<0, \alpha\in(\tilde{\alpha}_{Np},\tilde{\alpha}_{Np}+\Delta_2).\end{equation}
	which implies that the $\tilde{\alpha}_{Np}$ is a local maximizer of ${W({\alpha},Np)} - {W(r{\alpha},Np)}$. 
	
	Then, 	we prove that the local maximizer $\tilde{\alpha}_{Np}$ is unique, so that it is a global maximizer, and it satisfies \eqref{eq:optimal-alpha}. The second-order derivative of ${W(\alpha,Np)} - {W(r\alpha,Np)}$ with respect to $\alpha$ is given by
	\begin{multline}\label{eq:Wgap-derivative2}
		[{W(\alpha,Np)} - W(r\alpha,Np)]_{\alpha\alpha} \\
		= \int^{\overline{Q}}_{\underline{Q}}\left(e^{ N p(F(x)-1)\alpha}-r^2 e^{ N p(F(x)-1)r\alpha}\right)\\
		\times\left(N p(F(x)-1)\right)^2F(x)x^{-2}dx.
	\end{multline}
	According to First Mean Value Theorem for Definite Integrals, there exists a $\xi\in[\overline{Q},\underline{Q}]$ such that 
	\begin{multline}
		[{W(\alpha,Np)} - W(r\alpha,Np)]_{\alpha\alpha}  \\= N p(F(\xi)-1) \int^{\overline{Q}}_{\underline{Q}}N p(F(x)-1)F(x)x^{-2}\\  \times \bigg(e^{ N p(F(x)-1)\alpha}-r^2e^{ N p(F(x)-1)r\alpha}\bigg)dx.
	\end{multline}
	Recall that, at $\tilde{\alpha}_{Np}$, the first-order derivative $[{W(\tilde{\alpha}_{Np},Np)} - {W(r\tilde{\alpha}_{Np},Np)}]_{\alpha}=0$. By substituting it, the second order derivative at $\tilde{\alpha}_{Np}$ is given by
	\begin{multline}
		[{W(\tilde{\alpha}_{Np},Np)} - W(r\tilde{\alpha}_{Np},Np)]_{\alpha\alpha} \\
		= N p(F(\xi)-1) \int^{\overline{Q}}_{\underline{Q}}N p(F(x)-1)F(x)x^{-2} \\ \times e^{ N p(F(x)-1)r\tilde{\alpha}_{Np}} r(1-r)dx\leq 0.
	\end{multline}
	As a result of $[{W(\tilde{\alpha},Np)} - W(r\tilde{\alpha},Np)]_{\alpha\alpha} \leq 0$, the $\tilde{\alpha}_{Np}$ has to be unique. If it is not unique, there must exist at least one other $\hat{\alpha}\neq \tilde{\alpha}_{Np}$ satisfying $[{W(\hat{\alpha},Np)} - {W(r\hat{\alpha},Np)}]_{\alpha}=0$ such that $[{W(\hat{\alpha},Np)} - W(r\hat{\alpha},Np)]_{\alpha\alpha}>0$, which constradits $[{W(\hat{\alpha},Np)} - W(r\hat{\alpha},Np)]_{\alpha\alpha}\leq   0.$
	Hence, the $\tilde{\alpha}_{Np}$ is unique, and it is the global maximizer of ${W(\alpha,Np)} - {W(r{\alpha},Np)}$. This always implies that $[{W(\alpha,Np)} - {W(r{\alpha},Np)}]_{\alpha}> 0$ if  $\alpha<\tilde{\alpha}_{Np}$; and  $[{W(\alpha,Np)} - {W(r{\alpha},Np)}]_{\alpha}< 0$ if  $\alpha>\tilde{\alpha}_{Np}$.
	
	\emph{Claim (2):} We first prove that there exists a  $\tilde{p}$ such that $ [{W(1/r, N\tilde{p})} - {W(1, N\tilde{p})}]_{\alpha}=0$. We then show that the $\tilde{p}$ is unique, and satisfies \eqref{eq:optimal-p}.%  $[{W(1/r, Np)} - {W(1, Np)}]_{\alpha}>0$, if $p<\tilde{p}$, and $[{W(1/r, Np)} - {W(1, Np)}]_{\alpha}<0$ if $p>\tilde{p}$. %Note that this is equivalent to $\alpha^*=1/r$ if $p\leq\bar{p}$, and $\alpha^*=a^{\star}$ if $p>\bar{p}$. 	
	
	First, we prove that there exists a  $\tilde{p}$ such that $ [{W(1/r, N\tilde{p})} - {W(1, N\tilde{p})}]_{\alpha}=0$. The proof idea is similar as above. Checking the limitation of $[{W(1/r, Np)} - {W(1, Np)}]_{\alpha}$, we have 
	\begin{equation}
		\lim_{Np\rightarrow 0} [{W(1/r, Np)} - {W(1, Np)}]_{\alpha} >0,
	\end{equation}
	\begin{equation}
		\lim_{Np\rightarrow \infty} [{W(1/r, Np)} - {W(1, Np)}]_{\alpha} \rightarrow 0^-. 
	\end{equation}
	According to Intermediate Value Theorem, there exists at least a $\tilde{p}$ such that 
	\begin{multline}
		[{W(1/r, N\tilde{p})} - {W(1, N\tilde{p})}]_{\alpha}=	\int^{\overline{Q}}_{\underline{Q}} N\tilde{p}(F(x)-1)F(x)x^{-2}\\
		\times\left(e^{ N \tilde{p}(F(x)-1)/r}-re^{ N \tilde{p}(F(x)-1)}\right) dx=0\end{multline}
	According to First Mean Value Theorem for Definite Integrals, this equality can be represented as 
	\begin{equation}
		\int^{\overline{Q}}_{\underline{Q}}(F(x)-1)F(x) \left(e^{\frac{ N \tilde{p}(F(x)-1)}{r}}-re^{ N \tilde{p}(F(x)-1)}\right)dx=0.\end{equation}	
	
	Then, we prove that the $\tilde{p}$ is unique, and satisfies \eqref{eq:optimal-p}. Taking the first order derivative of $[{W(1/r, Np)} - {W(1, Np)}]_{\alpha}$ with respect to $Np$, we have 
	\begin{multline}
		[[{W(1/r, N{p})} - {W(1, N{p})}]_{\alpha}]_{Np} \\
		= \int^{\overline{Q}}_{\underline{Q}}\left(e^{ N p(F(x)-1)\frac{1}{r}}\frac{1}{r}-e^{ N p(F(x)-1)}r\right)\\
		\times Np(F(x)-1)^2F(x)x^{-2}dx.
	\end{multline}
	According to First Mean Value Theorem for Definite Integrals, there exists a $\xi\in[\overline{Q},\underline{Q}]$ such that 
	\begin{multline}
		[[{W(1/r, N{p})} - {W(1, N{p})}]_{\alpha}]_{Np} \\
		=N p(F(\xi)-1) \int^{\overline{Q}}_{\underline{Q}}(F(x)-1)F(x)x^{-2} \\ 
		\times 
		\left(e^{ N p(F(x)-1)\frac{1}{r}}\frac{1}{r} -e^{ N p(F(x)-1)}r\right) dx.
	\end{multline}
	Substituting ${W(1/r, N\tilde{p})} - {W(1, N\tilde{p})}]_{\alpha} = 0$, we have 
	\begin{multline}
		[[{W(1/r, N\tilde{p})} - {W(1, N\tilde{p})}]_{\alpha}]_{Np} \\
		=N \tilde{p}(F(\xi)-1) \int^{\overline{Q}}_{\underline{Q}}(F(x)-1)F(x)x^{-2} \\
		\times
		e^{ N \tilde{p}(F(x)-1)\frac{1}{r}}(\frac{1}{r} -1)dx\leq 0.
	\end{multline}
	Hence, similarly, as a result of $[[{W(1/r, N\tilde{p})} - {W(1, N\tilde{p})}]_{\alpha}]_{Np} \leq 0$, the $\tilde{p}$ has to be unique. If it is not unique, there must exist at least one other $\hat{p}\neq \tilde{p}$ satisfying $[{W(1/r, N\hat{p})} - {W(1, N\hat{p})}]_{\alpha}$ such that $[[{W(1/r, N\hat{p})} - {W(1, N\hat{p})}]_{\alpha}]_{Np}>0$, which contradicts $[[{W(1/r, N\hat{p})} - {W(1, N\hat{p})}]_{\alpha}]_{Np} \leq   0$. This shows the uniqueness of $\tilde{p}$, and shows that $\tilde{p}$ satisfies \eqref{eq:optimal-p}.
	
	Based on \emph{Claim (1)} and \emph{Claim (2)}, we now prove Lemma \ref{lemma:optimal-alpha}. Considering $\Delta W(r,{\alpha^{1C},p}) \triangleq W(\alpha^{1C},Np) - W(\alpha^{3C},Np)$ with $\alpha^{3C} = \min\{r\alpha,1\}$ under the following three cases:
	\begin{itemize}
		\item  $p<\tilde{p}$:  $[{W(1/r, Np)} - {W(1, Np)}]_{\alpha}>0$. According to \emph{Claim (1)}, this means that $\Delta W(r,{\alpha^{1C},p})$ is increasing in $\alpha^{1C}$ when $\alpha^{1C}<1/r$. Note that when $\alpha^{1C}\geq 1/r$, $\Delta W(r,{\alpha^{1C},p})$ is decreasing in $\alpha^{1C}$, as $W(\alpha^{1C},Np)$ is decreasing in $\alpha^{1C}$, and $W(\alpha^{3C},Np)$ is fixed. Hence, the optimal $\alpha^{1C}$ achieves at $\alpha^{1C}=\alpha^*_{Np} =1/r$. 
		\item  $p=\tilde{p}$: $[{W(1/r, Np)} - {W(1, Np)}]_{\alpha}= 0$, so  $\alpha^*_{Np} = 1/r$. 
		\item $p>\tilde{p}$:  $[{W(1/r, Np)} - {W(1, Np)}]_{\alpha}<0$. According to  \emph{Claim (1)},  $\alpha^*_{Np} = \tilde{\alpha}_{Np}\leq 1/r$. 
	\end{itemize}
\end{proof}
\begin{lemma}[Optimal $\alpha^{1C}$ and $Np$]\label{lemma:optimal-alpha-p}For any ratio $r\geq 1$ and distribution $f(x)$, the optimal $\alpha^{1C}$ and $Np$ that maximizes  $\Delta W(r,{\alpha^{1C},p})$ is as follows:
	\begin{equation}
		\alpha^{1C} = {\alpha}_{N\tilde{p}}^* = 1/r,
	\end{equation}
	and  $ \tilde{p}$  satisfies
	\begin{equation}
		\int^{\overline{Q}}_{\underline{Q}}(F(x)-1)F(x)\left(e^{ \frac{N \tilde{p}(F(x)-1)}{r}}-re^{ N \tilde{p}(F(x)-1)}\right) dx = 0.
	\end{equation}
\end{lemma}
\begin{proof}
	We will discuss the two cases $p\leq\tilde{p}$ and  $p\geq\tilde{p}$ one by one, and show that under each of these two cases,  $\Delta W(r,{\alpha}_{Np}^*,p)$ is maximized at $p=\tilde{p}$.
	
	Under $p\leq\tilde{p}$, according to Lemma \ref{lemma:optimal-alpha}, $\alpha^*_{Np} = 1/r$. Then, $\Delta W(r,{\alpha}_{Np}^*,Np)$ is given by 
	\begin{equation}
		\Delta W(r,{\alpha}_{Np}^*,Np) = {W(1/r, Np)} - {W(1, Np)}.
	\end{equation}
	Taking the first-order derivative of ${W(1/r, Np)} - {W(1, Np)}$ with respect to $Np$, we have 
	\begin{multline}
		[{W(1/r, Np)} - {W(1, Np)}]_{Np} = \int^{\overline{Q}}_{\underline{Q}}(F(x)-1)F(x)x^{-2} \\
		\times \left(e^{ N p(F(x)-1)\frac{1}{r}}\frac{1}{r} -e^{ N p(F(x)-1)}\right) dx\geq 0.
	\end{multline}
	This shows that $\Delta W(r,{\alpha}_{Np}^*,Np)$ is maximized at $p=\tilde{p}$.
	
	Under $p\leq\tilde{p}$,  $\alpha^*_{Np} = \tilde{\alpha}_{Np}\leq 1/r$, which is $\Delta W(r,{\alpha}_{Np}^*,Np) = W(\tilde{\alpha}_{Np},Np)- W(r\tilde{\alpha}_{Np},Np)$. Taking derivative of $\Delta W(r,{\alpha}_{Np}^*,Np)$ with respect to $Np$, we have 
	\begin{multline}
		[\Delta W(r,{\alpha}_{Np}^*,Np)]_{Np} = \tilde{\alpha}_{Np} \int^{\overline{Q}}_{\underline{Q}}(F(x)-1)F(x)x^{-2} \\
		\times \left(e^{ N p(F(x)-1)\tilde{\alpha}_{Np}}-re^{ N p(F(x)-1)\tilde{\alpha}_{Np} r}\right) dx.
	\end{multline}
	On the other hand, since that 
	\begin{multline}
		[\Delta W(r,{\alpha}_{Np}^*,Np)]_{\alpha} =  Np \int^{\overline{Q}}_{\underline{Q}}(F(x)-1)F(x)x^{-2} \\
		\times  \left(e^{ N p(F(x)-1)\tilde{\alpha}_{Np}}-re^{ N p(F(x)-1)\tilde{\alpha}_{Np} r}\right) dx = 0.
	\end{multline}
	Hence, $[\Delta W(r,{\alpha}_{Np}^*,Np)]_{Np} = 0$, so that $\Delta W(r,{\alpha}_{Np}^*,Np)$ is fixed as $Np$ changes. We can say that $p=\tilde{p}$ maximizes $\Delta W(r,{\alpha}_{Np}^*,Np)$.
	
	To sum up,  the $\alpha^{1C} = {\alpha}_{N\tilde{p}}^* = 1/r$ and the $\tilde{p}$  maximizes  $\Delta W(r,{\alpha^{1C},p})$.
\end{proof}

Based on Lemma \ref{lemma:optimal-alpha-p}, the $\alpha^{1C}  = 1/r$ and $p=\tilde{p}$ satisfying $\int^{\overline{Q}}_{\underline{Q}}(F(x)-1)F(x)\left(e^{{N \tilde{p}(F(x)-1)}/{r}}-re^{ N \tilde{p}(F(x)-1)}\right) dx = 0$ maximizes  $\Delta W(r,{\alpha^{1C},p})$. Through plugging in these $\alpha^{1C}$ and the $\tilde{p}$, we obtain the maximum energy reduction as in Theorem \ref{thm:maxenergy}.

\subsection{Truncated Normal Distribution}\label{app:trun}
Intuitively, a truncated normal distribution $f(x;\mu,\sigma,a,b) $ is a normal distribution $N(\mu,\sigma)$ but lies within range $(a,b)$. %: given a normal distribution $N(\mu,\sigma)$  and a range $(a,b)$, the corresponding truncated normal distribution $f(x;\mu,\sigma,a,b)$ remains the density of $N(\mu,\sigma)$  with range $(a,b)$, and adds additional density to the $N(\mu,\sigma)$  by allocating the densities outside $(a,b)$ proportional to the densities of $N(\mu,\sigma)$ within $(a,b)$. 
Formally, according to paper \cite{distribution}, the truncated distribution is defined as follows:
\begin{equation}
	f(x;\mu,\sigma,a,b) = \frac{\phi(\frac{x-\mu}{\sigma})}{\sigma\left(\Phi\left(\frac{b-\mu}{\sigma}\right)-\Phi\left(\frac{a-\mu}{\sigma}\right)\right)}, \forall x\in(a,b),
\end{equation}
where functions $\phi(\xi)$ and $\Phi(\xi) $ are given by
\begin{equation}
	\phi(\xi) = \frac{1}{\sqrt{2\pi}}e^{-\frac{1}{2}\xi^2},
\end{equation}
\begin{equation}
	\Phi(\xi) = \frac{1}{2}\left(1+\frac{2}{\sqrt{\pi}}\int^{\frac{\xi}{\sqrt{2}}}_{0}e^{-t^2}dt\right).
\end{equation}
%	The distribution range $(a,b)$ of  the parameters are shown in Table \ref{table:para}. 

\subsection{Analysis and Proof for Equation \eqref{eq:cache}}\label{app:eqcache}
The analysis for caching will be similar as it for communication/computation. We will first formulate the expected energy consumption, then introduce the analysis idea. In the random group $G(N,p)$, suppose each device joins the cooperative system with a probability $\alpha\in[0,1]$. Under  these case, retrieving each content will have an expected energy of $Z(\alpha,Np)$.\footnote{To clarify, if the requested content cannot be found in the cache of the devices in the cooperative system, the content will be downloaded by some devices. So the expected energy here corresponds to the downloading energy. In addition, under the homogeneous distribution settings in Section \ref{subsec:analysis-setting}, retrieving any content will have the same expected energy consumption, so we only need to study the expected energy consumption of a content.} Under the 1C model, let us denoted the probability that each device joins the caching sharing model as $\alpha^{1C}\in[0,1]$; under the 3C framework, the corresponding probability is $\alpha^{3C}=\min\{r\alpha^{1C},1\}$, where $r\geq 1$ is a coefficient reflecting the ratio of the increased cooperation opportunities. We will compute the energy reduction $\Delta Z(r,\alpha^{1C},Np) \triangleq Z(\alpha^{1C},Np)-Z(\alpha^{3C},Np)$ for $r\geq1$. 

In the random graph, for a device with degree $m$, the probability that $\hat{N}$ of his neighbors implement the content sharing  is $P(\hat{N}|m) = C_{m}^{\hat{N}}\alpha^{\hat{N}}(1-\alpha)^{m-\hat{N}}$, and the expected probability that the content has not been cached by these $\hat{N}$ devices and itself is $\left(1-\frac{\Mcache}{K}\right)^{\hat{N}+1}$. Taking the expectation over $\hat{N}=\{0,1,...,m\}$, the expected probability that the content has to be downloaded is given by 
\begin{equation}
	\hat{Z}_m(\alpha,Np) = \sum_{\hat{N}=0}^{m} P(\hat{N}|m) \left(1-\frac{\Mcache}{K}\right)^{\hat{N}+1}.
\end{equation}
Taking the expectation of $\hat{Z}(\alpha,m)$ over all degrees $m=\{0,1,...,\infty\}$,  the expected probability that the content has to be downloaded is 
\begin{equation}\label{eq:caching-energy}
	{Z}(\alpha,Np)  = (1-\frac{\Mcache}{K}) e^{-\alpha Np \frac{\Mcache}{K}},
\end{equation}
which is the Equation \eqref{eq:cache}. To clarify, we do not consider the sharing of downloading resources here, so the expected energy consumption of a content is the  expected probability that the content has to be downloaded multiplied by an expected downloading energy of a device. Normalizing the expected downloading energy to be one, \eqref{eq:caching-energy} is the  expected energy consumption of a content.

\subsection{Proof for Theorem \ref{thm:2}}\label{app:thrm2}
We first present two lemmas indicating the optimal $\alpha^{1C}$ and $p$  that maximizes $\Delta Z(r,{\alpha^{1C},p}) \triangleq {Z}(\alpha^{1C},Np) - {Z}(\alpha^{3C},Np) $ for any $r\geq 1$, then show the maximum energy reduction, i.e., $\max_{\alpha^{1C},Np} \Delta Z(r,{\alpha^{1C},Np})$.

\begin{lemma}[Optimal $\alpha^{1C}$ under a Particular $Np$]\label{lemma:caching-alpha}
	For any ratio $r$, connection probability $p$, and distribution $f(x)$, there exists a $\alpha^{1C}=\alpha^*_{Np}$ that maximizes the  reduction $\Delta Z(r,{\alpha^{1C},Np})$, i.e.,
	\begin{equation}
		\alpha^*_{Np} = \left\{\begin{array}{ll}
			\frac{1}{r}, & Np\leq \frac{r\ln r }{(r-1)\frac{\Mcache}{K}}\\
			\frac{\ln r}{(r-1)Np \frac{\Mcache}{K}}, & Np> \frac{r\ln r}{(r-1)\frac{\Mcache}{K}}
		\end{array}\right.
	\end{equation}
\end{lemma}
The proof idea is similar as it for Lemma \ref{lemma:optimal-alpha}, and we omit the details. Specifically, we first prove two claims. \emph{Claim (1):} for any $Np$, there exists a unique $\tilde{\alpha}_{Np}$ that maximizes ${Z({\alpha},Np)} - {Z(r{\alpha},Np)}$, and it satisfies
\begin{equation}\label{eq:optimal-alpha2}
	[{Z(\alpha,Np)} - {Z(r{\alpha},Np)}]_{\alpha}\left\{\begin{array}{ll}
		> 0,&   \alpha<\tilde{\alpha}_{Np}\\
		= 0,&   \alpha=\tilde{\alpha}_{Np}\\
		< 0,& \alpha>\tilde{\alpha}_{Np}\end{array}\right.,
\end{equation}
where 
\begin{equation}
	\tilde{\alpha}_{Np} = \frac{\ln r}{(r-1)Np \frac{\Mcache}{K}}.
\end{equation}
\emph{Claim (2):} there exists a unique $\tilde{p} $ that satisfies
\begin{equation}\label{eq:optimal-p2}
	[{Z(1/r, N\tilde{p})} - {Z(1, N\tilde{p})}]_{\alpha}\left\{\begin{array}{ll}
		>0,& p<\tilde{p}\\
		=0,& p=\tilde{p}\\
		<0,& p>\tilde{p}\\
	\end{array}
	\right.,
\end{equation}
where 
\begin{equation}
	\tilde{p}= \frac{r\ln r}{(r-1)\frac{\Mcache}{K}}.
\end{equation}
Then, using these two claims, we can prove Lemma \ref{lemma:caching-alpha}.

\begin{lemma}[Optimal $\alpha^{1C}$ and $\tilde{p}$]\label{lemma:optimal-alpha-p2} For any ratio $r\geq1$ and distribution $f(x)$, the optimal $\alpha^{1C}$ and $\tilde{p}$ that maximizes $\Delta Z(r,{\alpha^{1C},p})$ is as follows:
	\begin{equation}
		\alpha^{1C} = \alpha^*_{N\tilde{p}} = 1/r,~
		\tilde{p}= \frac{r\ln r}{(r-1)\frac{\Mcache}{K}}.
	\end{equation}	
\end{lemma}
The proof idea is similar as it for Lemma \ref{lemma:optimal-alpha-p}, and we omit the details. Specifically, we can check each of the case $p\leq \tilde{p}$ and $p\geq \tilde{p}$, and show that $p\tilde{p}$ maximizes $\Delta Z(r,{\alpha^*_{Np},p})$.

Based on Lemma \ref{lemma:optimal-alpha-p2}, $\alpha^{1C} = 1/r$ and ${p}= {r\ln r}/((r-1){\Mcache}/{K})$ maximizes $\Delta Z(r,{\alpha^{1C},p})$. Through plugging in the $\alpha^{1C}$ and the ${p}$, we obtain Theorem \ref{thm:2}.

\subsection{Simulation and Performance}\label{app:exp}
We consider a scenario of with a set of $N$ devices, who form pair-wise connections with probability $p$. Each device has one task to execute. For each experiment, we perform 100 times (if not specified) and show the average results. In each time of an experiment, we randomly generate the parameters of the device and task models, including devices' capacities, tasks' demands (computation requirement and content sizes), and energy consumption coefficients.
These parameters follows truncated normal distribution \cite{distribution}, which is a normal distribution but lies within range $(a,b)$, i.e., for $x\in(a,b)$,
\begin{equation}
	f(x;\mu,\sigma,a,b) = \frac{\phi(\frac{x-\mu}{\sigma})}{\sigma\left(\Phi\left(\frac{b-\mu}{\sigma}\right)-\Phi\left(\frac{a-\mu}{\sigma}\right)\right)},
\end{equation}
where functions $\phi(\xi)$ and $\Phi(\xi) $ are given by
\begin{equation}
	\phi(\xi) = \frac{1}{\sqrt{2\pi}}e^{-\frac{1}{2}\xi^2},
\end{equation}
\begin{equation}
	\Phi(\xi) = \frac{1}{2}\left(1+\frac{2}{\sqrt{\pi}}\int^{\frac{\xi}{\sqrt{2}}}_{0}e^{-t^2}dt\right).
\end{equation}
The distribution range $(a,b)$ of the parameters are shown in Table \ref{table:para}. The relative values of the maximum downloading, uploading, and D2D transmission capacities are based on references \cite{down-cap} and \cite{wifi-direct-cap}, and the relative values of the maximum energy per time are based on paper \cite{energy} and \cite{up-energy}. We pick the same value for the maximum computation capacity and demand, so that performing a computation subtask is one second on average. In addition, for each of these parameters, we set $\mu=(a+b)/2$ and $\sigma=1.0$ (if not specified), under which the distribution is similar to a uniform distribution within range $(a,b)$. In addition, the content sizes are set to be one. Each device caches a random number of contents in its local cache and requests a random number of content input and content output. The uploading contents and caching contents are randomly selected from the content output. Moreover, each tasks' delay constraints are randomly generated parameters that can guarantee that noncooperation decision is always in the feasible set of Problem $(\text{OPT})$.
\begin{table}[t]
	\caption{Parameter Settings in Experiments}\label{table:para}
	\begin{center}
		\begin{tabular}{cccccc}
			\hline
			Parameter & $(a,b)$ & Parameter & $(a,b)$ &Parameter & $(a,b)$ \\
			\hline
			$\Qdown$ &$(0,10)$ &$\Qup$&$(0,4)$&$\Qd2d$ &$(0,50)$\\
			$\Qcpu$ &$(0,10)$&$\Dcpu$& $(0,10)$ &  $\edown$ & $(0,2.8)$\\
			$\ecpu$&$(0,1.2)$&$\eup$&$(0,2.8)$&$\ed2d$&$(0,0.8)$\\
			\hline			
		\end{tabular}
	\end{center}
\end{table}

\end{document}